\documentclass{lmcs}


\usepackage{booktabs}
\usepackage{tikz-cd}
\usepackage{amsmath}
\usepackage{mathtools}
\usepackage{amssymb}
\usepackage{bussproofs}
\newcommand{\diag}{\mathrm{diag}}
\DeclareMathOperator{\id}{id}
\DeclareMathOperator{\Ob}{Ob}
\DeclareMathOperator{\disc}{disc}
\DeclareMathOperator{\Str}{Str}
\DeclareMathOperator{\Out}{Out}

\newcommand{\Ctx}{\mathbf{Ctx}}
\newcommand{\Arch}{\mathbf{Arch}}

\newcommand{\TZ}{\mathsf{TZ}}
\newcommand{\Diag}{\mathrm{Diag}}
\newcommand{\Sig}{\Sigma}

\newcommand{\sem}[1]{[\![#1]\!]}
\newcommand{\zone}[2]{#1\!:\!#2}
\newcommand{\out}{\mathsf{o}}

\newcommand{\W}{\mathcal{W}}
\newcommand{\C}{\mathcal{C}}
\newcommand{\tensor}{\otimes}
\newcommand{\unit}{I}
\newcommand{\ctx}[1]{\langle #1 \rangle}
\newcommand{\Lic}{\mathsf{Lic}}

\begin{document}

\title[Architectures to subexponential logic]{From Categorized Neural Architectures\texorpdfstring{\\}{} to Subexponential Proof Theory}

\author[C.~Ram\'{\i}rez Ovalle]{Carlos  Ram\'{\i}rez Ovalle\lmcsorcid{0000-0002-4254-5610}}[a]

\address{Department of Natural Sciences and Mathematics, Pontificia Universidad Javeriana Cali, Colombia}
\email{carlosovalle@javerianacali.edu.co}

\begin{abstract}
We study a resource-sensitive fragment of the problem of extracting a logical discipline from a class of neural architectures by passing through categorization. The starting point is not a pre-existing logic but a category of zone-labelled parametrised blocks together with a disciplined record of which forms of copying, discarding, and zone coercion are architecturally licensed. From this categorized architecture we read off a subexponential signature and then define a tensorial sequent calculus whose structural rules are indexed by the extracted zones. The paper proves three kinds of results. First, the resulting architectural category is symmetric monoidal. Second, the extracted proof system admits cut elimination. Third, derivations are sound with respect to the licensed categorical diagrams generated by the architectural discipline. The outcome is a theorem-bearing core of the architecture-to-category-to-logic programme: subexponential structure is not postulated in advance but read from categorical data encoding differentiated memory and context behaviour.
\end{abstract}

\maketitle

\section{Introduction}\label{sec:intro}

The use of category theory in machine learning has progressed well beyond analogy. Compositional accounts of supervised learning and backpropagation show that parametrised learning systems admit categorical semantics in which composition and update can be studied functorially \cite{FongSpivakTuyeras2019}. More recently, categorical deep learning has been proposed as a broad algebraic language for architectural composition, reparametrisation, and weight sharing \cite{GavranovicEtAl2024}. On the other hand, categorical logic studies how structural features of categories give rise to doctrines, internal languages, and proof systems \cite{Jacobs1999,MacLane1998}. These two lines suggest a natural question: can one begin from a suitably organized family of neural architectures, pass to a categorical description of their resource behaviour, and then extract an appropriate logic from that categorical structure?

The present article develops a first theorem-bearing answer in a resource-sensitive setting. The emphasis is important. Our aim is not to begin with a fully formed logic and then search for a semantics. Instead, we start with a zone-labelled architectural formalism. Some forms of information are reusable, some discardable, some linear, and some coercible from one zone to another. These permissions are encoded categorically. Only after the categorical organization is fixed do we read off a subexponential signature and define the corresponding proof theory.

This order of construction can be displayed schematically as
\[
\text{architecture family}
\longrightarrow
\text{categorized architecture}
\longrightarrow
\text{subexponential signature}
\longrightarrow
\text{proof theory}.
\]
The final step is then to show that the extracted proof system is sound in the original categorized architecture. In this sense the category plays two roles: it is first an abstraction of architectural structure and only afterwards a semantics validating the logic extracted from that structure.

For broader background on categorical semantics of linear logic, including the passage
from linear structure to categorical models and the treatment of exponentials, see
\cite{See89,Bar91,HS03,Mel09}.

The resulting logic is not intended to describe an entire transformer or a complete contemporary learning stack. That would require a richer treatment of nonlinearities, internal homs, and likely a linear--nonlinear adjunction in the sense of Benton \cite{Benton1995}. The goal here is more focused. We isolate the part of the logical discipline already forced by differentiated structural behaviour. This leads naturally to a tensorial fragment with zone-indexed weakening and contraction, that is, to a subexponential proof theory in which structural permissions depend on zones \cite{NigamMiller2009,XavierOlartePimentel2022}. The proof-theoretic stance is therefore deliberately modest but mathematically crisp.

The main novelty of the paper is not the mere use of subexponentials. Subexponential proof theory and semantics are by now well established, and recent work has pushed their categorical analysis significantly further \cite{Rogozin2025}. What is new here is the direction of travel: the subexponential data are not fixed beforehand as logical syntax, but extracted from a categorized resource discipline motivated by architectures with differentiated memory and context. The key theorem is therefore an extraction-and-soundness result rather than a purely semantic one.

The contributions of the article are the following.
\begin{itemize}
\item We define a category $\Arch_Z$ of zone-labelled parametrised architectural blocks and prove that it is symmetric monoidal.
\item We introduce the notion of a \emph{coherent zone discipline} on $\Arch_Z$, consisting of licensed zone coercions, discards, and diagonals.
\item From such a disciplined category we extract a subexponential signature $\Sig_{\Lic}$ and an associated tensorial sequent calculus $\TZ_{\Sig_{\Lic}}$.
\item We prove cut elimination for $\TZ_{\Sig_{\Lic}}$.
\item We prove soundness of $\TZ_{\Sig_{\Lic}}$ with respect to the class of licensed diagrams generated by the categorized architecture.
\end{itemize}

The article should therefore be read as a foundational first step in the programme
\[
\text{neural architecture} \to \text{category} \to \text{logic},
\]
with particular attention to resource-sensitive behaviour. The missing connective is linear implication. We explain below why its absence does not undercut the present contribution: the current theorems already identify the structural fragment of the logic that is forced by the architecture, and that fragment is the one most directly tied to differentiated resource usage.

At a high level, the main theorem-bearing content of the paper can be summarized as follows:
from a coherently disciplined category of zone-labelled architectural blocks one can
canonically reconstruct a subexponential signature, define the associated tensorial
proof theory, and interpret every derivation back into the licensed categorical diagrams.

Accordingly, the present article should not be read as a logic of all neural architectures,
nor as a full logic of transformers, but as a theorem-bearing identification of the structural
resource fragment already forced by categorized architectural permissions.

The individual results of Sections~2--4 may thus be read as the categorical, syntactic,
and semantic components of a single extraction theorem.

Section~4 closes with explicit examples, including a two-block neural architecture showing
concretely how differentiated input roles induce the extracted substructural discipline.

\section{Categorized architectures}\label{sec:arch}

\subsection{Typed resource contexts and the architectural category}

We now introduce the categorical object that mediates between architecture and logic. 
The purpose of this subsection is to extract, from a resource-sensitive architectural description, 
a category whose objects are zone-labelled typed contexts and whose morphisms are parametrised computational blocks. 
This construction is the first step in the chain
\[
\text{architecture}\;\longrightarrow\;\text{categorisation}\;\longrightarrow\;\text{subexponential discipline}\;\longrightarrow\;\text{proof theory}.
\]
At this stage, only the compositional structure of contexts and blocks is relevant. 
Accordingly, we work over a category with finite products and make no further assumptions.

Let $\Ctx$ be a category with finite products, and let $Z$ be a fixed set of zone labels.

\begin{defi}[Typed resource context]\label{def:typed-context}
A \emph{typed resource context} is a finite list
\[
\Gamma=((z_1,A_1),\dots,(z_n,A_n)),
\]
where each $z_i\in Z$ and each $A_i\in\Ob(\Ctx)$. Its \emph{carrier} is the object
\[
\sem{\Gamma}:=A_1\times\cdots\times A_n,
\]
with the convention that the carrier of the empty list is the terminal object $1$ of $\Ctx$. 
The empty context is denoted by $\emptyset$.

If
\[
\Gamma=((z_1,A_1),\dots,(z_n,A_n))
\quad\text{and}\quad
\Delta=((w_1,B_1),\dots,(w_m,B_m)),
\]
their \emph{tensor concatenation} is the context
\[
\Gamma\tensor\Delta
:=
((z_1,A_1),\dots,(z_n,A_n),(w_1,B_1),\dots,(w_m,B_m)).
\]
\end{defi}

The notation $\Gamma\tensor\Delta$ is justified by the fact that concatenation of contexts is implemented on carriers by the product operation in $\Ctx$.

\begin{lem}\label{lem:carrier-tensor}
For all typed resource contexts $\Gamma$ and $\Delta$, there is a canonical isomorphism
\[
\mu_{\Gamma,\Delta}\colon \sem{\Gamma\tensor\Delta}\xrightarrow{\cong}\sem{\Gamma}\times\sem{\Delta}.
\]
\end{lem}

\begin{proof}
Write
\[
\Gamma=((z_1,A_1),\dots,(z_n,A_n)),
\qquad
\Delta=((w_1,B_1),\dots,(w_m,B_m)).
\]
By definition,
\[
\sem{\Gamma\tensor\Delta}
=
A_1\times\cdots\times A_n\times B_1\times\cdots\times B_m,
\]
whereas
\[
\sem{\Gamma}\times\sem{\Delta}
=
(A_1\times\cdots\times A_n)\times(B_1\times\cdots\times B_m).
\]
The universal property of finite products yields a unique morphism
\[
\mu_{\Gamma,\Delta}\colon \sem{\Gamma\tensor\Delta}\to \sem{\Gamma}\times\sem{\Delta}
\]
whose first component is the tuple of projections onto $A_1,\dots,A_n$ and whose second component is the tuple of projections onto $B_1,\dots,B_m$.

Similarly, there is a unique morphism
\[
\nu_{\Gamma,\Delta}\colon \sem{\Gamma}\times\sem{\Delta}\to \sem{\Gamma\tensor\Delta}
\]
whose components are the evident projections onto the factors
\[
A_1,\dots,A_n,B_1,\dots,B_m.
\]
The composites $\nu_{\Gamma,\Delta}\circ\mu_{\Gamma,\Delta}$ and $\mu_{\Gamma,\Delta}\circ\nu_{\Gamma,\Delta}$ are identities because they have the same projections as the corresponding identity morphisms. Hence $\mu_{\Gamma,\Delta}$ is an isomorphism, with inverse $\nu_{\Gamma,\Delta}$.\qedhere
\end{proof}

For $z\in Z$ and $A\in\Ob(\Ctx)$, we write
\[
\ctx{z,A}
\]
for the one-variable context $((z,A))$.

\begin{defi}[Parametrised resource block]\label{def:block}
Let $\Gamma$ and $\Delta$ be typed resource contexts. A \emph{parametrised resource block}
\[
(P,\varphi)\colon \Gamma\to\Delta
\]
consists of an object $P\in\Ob(\Ctx)$ together with a morphism
\[
\varphi\colon P\times\sem{\Gamma}\longrightarrow\sem{\Delta}
\]
in $\Ctx$.
\end{defi}

The parameter object is part of the presentation of a block, but not part of its essential behaviour. 
Two parametrised presentations should therefore be identified whenever they differ only by an isomorphic reparametrisation.

\begin{defi}[Isomorphism of parametrised blocks]\label{def:block-iso}
Let $(P,\varphi)$ and $(P',\varphi')$ be parametrised resource blocks $\Gamma\to\Delta$. 
An \emph{isomorphism of parametrised blocks}
\[
u\colon (P,\varphi)\xrightarrow{\cong}(P',\varphi')
\]
is an isomorphism $u\colon P\to P'$ in $\Ctx$ such that
\[
\varphi'\circ(u\times\id_{\sem{\Gamma}})=\varphi.
\]
We write
\[
(P,\varphi)\sim(P',\varphi')
\]
when such an isomorphism exists.
\end{defi}

\begin{lem}\label{lem:block-iso-equiv}
For fixed contexts $\Gamma$ and $\Delta$, the relation $\sim$ on parametrised resource blocks $\Gamma\to\Delta$ is an equivalence relation.
\end{lem}

\begin{proof}
Reflexivity is immediate, since
\[
\varphi\circ(\id_P\times\id_{\sem{\Gamma}})=\varphi.
\]

For symmetry, assume that
\[
u\colon (P,\varphi)\xrightarrow{\cong}(P',\varphi')
\]
is an isomorphism of parametrised blocks. Then
\[
\varphi'\circ(u\times\id_{\sem{\Gamma}})=\varphi.
\]
Composing on the right with $u^{-1}\times\id_{\sem{\Gamma}}$ gives
\[
\varphi\circ(u^{-1}\times\id_{\sem{\Gamma}})=\varphi',
\]
so
\[
u^{-1}\colon (P',\varphi')\xrightarrow{\cong}(P,\varphi)
\]
is again an isomorphism of parametrised blocks.

For transitivity, suppose
\[
u\colon (P,\varphi)\xrightarrow{\cong}(P',\varphi'),
\qquad
v\colon (P',\varphi')\xrightarrow{\cong}(P'',\varphi'').
\]
Then
\[
\varphi'\circ(u\times\id_{\sem{\Gamma}})=\varphi,
\qquad
\varphi''\circ(v\times\id_{\sem{\Gamma}})=\varphi'.
\]
Substituting the second equation into the first yields
\[
\varphi''\circ(v\times\id_{\sem{\Gamma}})\circ(u\times\id_{\sem{\Gamma}})=\varphi.
\]
Since the product is functorial,
\[
(v\times\id_{\sem{\Gamma}})\circ(u\times\id_{\sem{\Gamma}})
=
(v\circ u)\times\id_{\sem{\Gamma}},
\]
hence
\[
\varphi''\circ((v\circ u)\times\id_{\sem{\Gamma}})=\varphi.
\]
Therefore
\[
v\circ u\colon (P,\varphi)\xrightarrow{\cong}(P'',\varphi''),
\]
and $\sim$ is transitive.\qedhere
\end{proof}

\begin{defi}[The architectural category $\Arch_Z$]\label{def:arch}
The objects of $\Arch_Z$ are typed resource contexts.

A morphism
\[
[\![P,\varphi]\!]\colon \Gamma\to\Delta
\]
is a $\sim$-equivalence class of parametrised resource blocks $(P,\varphi)\colon \Gamma\to\Delta$.

Given morphisms
\[
[\![P,\varphi]\!]\colon \Gamma\to\Delta,
\qquad
[\![Q,\psi]\!]\colon \Delta\to\Theta,
\]
their composite is defined by
\[
[\![Q,\psi]\!]\circ[\![P,\varphi]\!]
:=
[\![Q\times P,\,
\psi\circ(\id_Q\times\varphi)\circ\alpha^{-1}_{Q,P,\Gamma}]\!],
\]
where
\[
\alpha_{Q,P,\Gamma}\colon
Q\times(P\times\sem{\Gamma})
\xrightarrow{\cong}
(Q\times P)\times\sem{\Gamma}
\]
is the canonical associativity isomorphism of the finite product.

The identity on $\Gamma$ is the class of
\[
(1,\pi_2)\colon \Gamma\to\Gamma,
\]
where $\pi_2\colon 1\times\sem{\Gamma}\to\sem{\Gamma}$ is the second projection.
\end{defi}

This composition is the evident sequential execution of parametrised blocks. 
The first block consumes the input context, and its output is fed to the second block; the composite parameter object simply remembers both parameter choices.

\begin{figure}[t]
\[
\begin{tikzcd}[column sep=large,row sep=large]
(Q\times P)\times \sem{\Gamma}
  \arrow[r,"\alpha^{-1}_{Q,P,\Gamma}"]
  \arrow[rrr,bend right=18,
    "\psi\circ(\id_Q\times\varphi)\circ\alpha^{-1}_{Q,P,\Gamma}"']
&
Q\times(P\times \sem{\Gamma})
  \arrow[r,"\id_Q\times\varphi"]
&
Q\times \sem{\Delta}
  \arrow[r,"\psi"]
&
\sem{\Theta}
\end{tikzcd}
\]
\caption{Sequential composition in $\Arch_Z$.}
\label{fig:composition}
\end{figure}

\begin{prop}\label{prop:category}
The data of \autoref{def:arch} define a category $\Arch_Z$.
\end{prop}

\begin{proof}
We first prove that composition is well defined on equivalence classes.

Let
\[
(P,\varphi)\sim(P',\varphi'),
\qquad
(Q,\psi)\sim(Q',\psi').
\]
Choose isomorphisms
\[
u\colon P\xrightarrow{\cong}P',
\qquad
v\colon Q\xrightarrow{\cong}Q'
\]
such that
\[
\varphi'\circ(u\times\id_{\sem{\Gamma}})=\varphi,
\qquad
\psi'\circ(v\times\id_{\sem{\Delta}})=\psi.
\]
We must show that the blocks
\[
(Q\times P,\,
\psi\circ(\id_Q\times\varphi)\circ\alpha^{-1}_{Q,P,\Gamma})
\]
and
\[
(Q'\times P',\,
\psi'\circ(\id_{Q'}\times\varphi')\circ\alpha^{-1}_{Q',P',\Gamma})
\]
are equivalent. The candidate isomorphism is
\[
v\times u\colon Q\times P\xrightarrow{\cong}Q'\times P'.
\]
Using functoriality of products and naturality of $\alpha^{-1}$, we compute
\begin{align*}
&\psi'\circ(\id_{Q'}\times\varphi')\circ\alpha^{-1}_{Q',P',\Gamma}
\circ((v\times u)\times\id_{\sem{\Gamma}}) \\
&\qquad =
\psi'\circ(\id_{Q'}\times\varphi')
\circ(v\times(u\times\id_{\sem{\Gamma}}))
\circ\alpha^{-1}_{Q,P,\Gamma} \\
&\qquad =
\psi'\circ(v\times\id_{\sem{\Delta}})
\circ(\id_Q\times\varphi)
\circ\alpha^{-1}_{Q,P,\Gamma} \\
&\qquad =
\psi\circ(\id_Q\times\varphi)\circ\alpha^{-1}_{Q,P,\Gamma}.
\end{align*}
Hence composition is well defined.

We next verify associativity. Let
\[
[\![P,\varphi]\!]\colon \Gamma\to\Delta,\qquad
[\![Q,\psi]\!]\colon \Delta\to\Theta,\qquad
[\![R,\chi]\!]\colon \Theta\to\Lambda.
\]
The two triple composites are represented, respectively, by
\[
((R\times Q)\times P,\;
\chi\circ(\id_R\times\psi)\circ\alpha^{-1}_{R,Q,\Theta}
\circ(\id_{R\times Q}\times\varphi)\circ\alpha^{-1}_{R\times Q,P,\Gamma})
\]
and
\[
(R\times(Q\times P),\;
\chi\circ(\id_R\times(\psi\circ(\id_Q\times\varphi)\circ\alpha^{-1}_{Q,P,\Gamma}))
\circ\alpha^{-1}_{R,Q\times P,\Gamma}).
\]
These representatives are equivalent via the canonical associativity isomorphism
\[
\alpha_{R,Q,P}\colon R\times(Q\times P)\xrightarrow{\cong}(R\times Q)\times P.
\]
Indeed, both composites describe the same iterated map obtained by reassociating
\[
R\times Q\times P\times\sem{\Gamma}
\]
so as to apply first $\varphi$, then $\psi$, then $\chi$. Since canonical maps between finite products are uniquely determined by their projections, the two resulting morphisms coincide.

Finally, we check the identity laws. Let
\[
[\![P,\varphi]\!]\colon \Gamma\to\Delta.
\]
The composite with the identity on the left is represented by
\[
(1\times P,\;
\pi_2\circ(\id_1\times\varphi)\circ\alpha^{-1}_{1,P,\Gamma}),
\]
which is equivalent to $(P,\varphi)$ via the canonical isomorphism
\[
\lambda_P\colon 1\times P\xrightarrow{\cong}P.
\]
Similarly, the composite with the identity on the right is represented by
\[
(P\times 1,\;
\varphi\circ(\id_P\times\pi_2)\circ\alpha^{-1}_{P,1,\Gamma}),
\]
which is equivalent to $(P,\varphi)$ via the canonical isomorphism
\[
\rho_P\colon P\times 1\xrightarrow{\cong}P.
\]
Therefore the stated identities are indeed units for composition, and $\Arch_Z$ is a category.\qedhere
\end{proof}

We next place blocks side by side. On contexts this is just concatenation; on morphisms it is parallel composition, obtained by keeping both parameter objects and letting each block see only its own part of the input carrier.

\begin{defi}[Tensor product on $\Arch_Z$]\label{def:tensor-arch}
On objects, the tensor product is context concatenation:
\[
\Gamma\tensor\Delta.
\]

On morphisms, let
\[
[\![P,\varphi]\!]\colon \Gamma\to\Gamma',
\qquad
[\![Q,\psi]\!]\colon \Delta\to\Delta'.
\]
Define
\[
[\![P,\varphi]\!]\tensor[\![Q,\psi]\!]
:=
[\![P\times Q,\varphi\tensor\psi]\!],
\]
where
\[
\varphi\tensor\psi\colon
(P\times Q)\times\sem{\Gamma\tensor\Delta}\longrightarrow \sem{\Gamma'\tensor\Delta'}
\]
is the composite
\[
(P\times Q)\times\sem{\Gamma\tensor\Delta}
\xrightarrow{\id\times\mu_{\Gamma,\Delta}}
(P\times Q)\times(\sem{\Gamma}\times\sem{\Delta})
\xrightarrow{\sigma}
(P\times\sem{\Gamma})\times(Q\times\sem{\Delta})
\xrightarrow{\varphi\times\psi}
\sem{\Gamma'}\times\sem{\Delta'}
\xrightarrow{\mu^{-1}_{\Gamma',\Delta'}}
\sem{\Gamma'\tensor\Delta'}.
\]
Here $\sigma$ denotes the canonical reassociation and symmetry isomorphism of finite products.
The monoidal unit is the empty context $\emptyset$.
\end{defi}

\begin{lem}\label{lem:tensor-well-defined}
The tensor product of \autoref{def:tensor-arch} is well defined on equivalence classes of parametrised resource blocks.
\end{lem}

\begin{proof}
Assume
\[
(P,\varphi)\sim(P',\varphi'),
\qquad
(Q,\psi)\sim(Q',\psi').
\]
Choose isomorphisms
\[
u\colon P\xrightarrow{\cong}P',
\qquad
v\colon Q\xrightarrow{\cong}Q'
\]
such that
\[
\varphi'\circ(u\times\id_{\sem{\Gamma}})=\varphi,
\qquad
\psi'\circ(v\times\id_{\sem{\Delta}})=\psi.
\]
We show that
\[
(P\times Q,\varphi\tensor\psi)\sim(P'\times Q',\varphi'\tensor\psi').
\]
The relevant isomorphism is $u\times v\colon P\times Q\to P'\times Q'$.

Expanding the definition of $\varphi\tensor\psi$ and $\varphi'\tensor\psi'$, and using the naturality of the canonical product isomorphisms involved, it suffices to verify commutativity of
\[
\begin{tikzcd}[column sep=huge,row sep=large]
(P\times\sem{\Gamma})\times(Q\times\sem{\Delta})
  \arrow[r,"\varphi\times\psi"]
  \arrow[d,"(u\times\id)\times(v\times\id)"']
&
\sem{\Gamma'}\times\sem{\Delta'}
  \arrow[d,equal]
\\
(P'\times\sem{\Gamma})\times(Q'\times\sem{\Delta})
  \arrow[r,"\varphi'\times\psi'"']
&
\sem{\Gamma'}\times\sem{\Delta'}.
\end{tikzcd}
\]
But this square commutes because
\[
\varphi'\circ(u\times\id_{\sem{\Gamma}})=\varphi
\qquad\text{and}\qquad
\psi'\circ(v\times\id_{\sem{\Delta}})=\psi.
\]
Hence
\[
(\varphi'\tensor\psi')\circ((u\times v)\times\id)
=
\varphi\tensor\psi,
\]
as required.\qedhere
\end{proof}

\begin{prop}\label{prop:symm-monoidal}
The category $\Arch_Z$, equipped with the tensor product of \autoref{def:tensor-arch} and unit object $\emptyset$, is a symmetric monoidal category.
\end{prop}

\begin{proof}
The associator, left and right unitors, and symmetry of $\Arch_Z$ are induced by the corresponding canonical isomorphisms of finite products in $\Ctx$, transported along the carrier isomorphisms of \autoref{lem:carrier-tensor}. More precisely, each of these structural morphisms is represented by a parametrised block with parameter object $1$ and underlying map the corresponding canonical product isomorphism.

It remains to verify the coherence axioms. Since all structural morphisms arise from canonical reassociations, symmetries, and unit isomorphisms of finite products, the pentagon, triangle, and hexagon identities in $\Arch_Z$ reduce to the usual coherence identities for monoidal categories. These hold by Mac Lane's coherence theorem; see \cite[Chapter~VII]{MacLane1998}. Therefore $\Arch_Z$ is a symmetric monoidal category.\qedhere
\end{proof}

\subsection{Coherent zone disciplines and extracted subexponential data}

The category $\Arch_Z$ encodes typed contexts and parametrised blocks, together with their sequential and parallel composition. 
This, however, is still not enough to recover a logical discipline. 
Indeed, because the ambient category $\Ctx$ has finite products, terminal and diagonal maps are always available at the level of carriers. 
If all such maps were admitted indiscriminately as structural, then the intended resource sensitivity would disappear. 
What must therefore be specified is not the mere existence of structural maps in $\Ctx$, but rather which among them are to be regarded as \emph{architecturally licensed}.

The role of a coherent zone discipline is precisely to supply this information. 
It enriches the architectural category with three kinds of data: 
a preorder of zone refinement, coercions between refined zones, and designated structural families expressing where discard and duplication are allowed. 
This is the point at which the passage
\[
\text{architecture}\;\longrightarrow\;\text{categorisation}\;\longrightarrow\;\text{subexponential structure}
\]
becomes explicit.

We begin with a piece of notation that allows us to speak precisely about naturality in the carrier object.

\begin{defi}[The one-zone embedding]\label{def:one-zone-embedding}
For each zone $z\in Z$, define a mapping
\[
J_z\colon \Ctx\longrightarrow \Arch_Z
\]
on objects by
\[
J_z(A):=\ctx{z,A},
\]
and on morphisms by
\[
J_z(f):=[\![1,f\circ \pi_2]\!]\colon \ctx{z,A}\longrightarrow \ctx{z,B}
\]
for every morphism $f\colon A\to B$ in $\Ctx$, where
\[
\pi_2\colon 1\times A\to A
\]
is the second projection.
\end{defi}

\begin{lem}\label{lem:Jz-functor}
For each $z\in Z$, the assignment $J_z$ of \autoref{def:one-zone-embedding} defines a functor
\[
J_z\colon \Ctx\to\Arch_Z.
\]
\end{lem}

\begin{proof}
Let $A\in\Ob(\Ctx)$. Then
\[
J_z(\id_A)=[\![1,\id_A\circ\pi_2]\!].
\]
Since
\[
\id_A\circ\pi_2=\pi_2\colon 1\times A\to A,
\]
this is exactly the identity morphism on $\ctx{z,A}$ in $\Arch_Z$, by \autoref{def:arch}. Hence
\[
J_z(\id_A)=\id_{J_z(A)}.
\]

Now let
\[
f\colon A\to B,
\qquad
g\colon B\to C
\]
be morphisms in $\Ctx$. By definition,
\[
J_z(f)=[\![1,f\circ\pi_2]\!],
\qquad
J_z(g)=[\![1,g\circ\pi_2]\!].
\]
Their composite in $\Arch_Z$ is represented by
\[
(1\times 1,\; (g\circ\pi_2)\circ(\id_1\times(f\circ\pi_2))\circ\alpha^{-1}),
\]
where
\[
\alpha\colon 1\times(1\times A)\xrightarrow{\cong}(1\times 1)\times A
\]
is the canonical associativity isomorphism. 
Using the canonical identifications
\[
1\times 1\cong 1
\qquad\text{and}\qquad
1\times A\cong A,
\]
the displayed morphism is canonically equal to
\[
g\circ f\circ \pi_2\colon 1\times A\to C.
\]
Therefore
\[
J_z(g)\circ J_z(f)=[\![1,(g\circ f)\circ\pi_2]\!]=J_z(g\circ f).
\]
Thus $J_z$ preserves identities and composition, and is therefore a functor.\qedhere
\end{proof}

We may now formulate the structural data that refine the bare category $\Arch_Z$.

\begin{defi}[Coherent zone discipline]\label{def:discipline}
A \emph{coherent zone discipline} $\Lic$ on $\Arch_Z$ consists of the following data.

\begin{enumerate}
\item A preorder $\preceq$ on $Z$.

\item For each pair $z\preceq z'$ and each object $A\in\Ob(\Ctx)$, a morphism
\[
\chi_{z,z',A}\colon J_z(A)\longrightarrow J_{z'}(A),
\]
called a \emph{zone coercion}, such that:

\begin{enumerate}
\item for every morphism $f\colon A\to B$ in $\Ctx$, the square
\[
\begin{tikzcd}[column sep=large,row sep=large]
J_z(A) \arrow[r,"\chi_{z,z',A}"] \arrow[d,"J_z(f)"'] &
J_{z'}(A) \arrow[d,"J_{z'}(f)"] \\
J_z(B) \arrow[r,"\chi_{z,z',B}"'] &
J_{z'}(B)
\end{tikzcd}
\]
commutes;

\item for every $A\in\Ob(\Ctx)$,
\[
\chi_{z,z,A}=\id_{J_z(A)};
\]

\item whenever $z\preceq z'\preceq z''$, one has
\[
\chi_{z',z'',A}\circ\chi_{z,z',A}=\chi_{z,z'',A}
\]
for every $A\in\Ob(\Ctx)$.
\end{enumerate}

\item Two upward-closed subsets
\[
\W_{\Lic},\C_{\Lic}\subseteq Z.
\]

\item For every $z\in\W_{\Lic}$ and every $A\in\Ob(\Ctx)$, a morphism
\[
\omega_{z,A}\colon J_z(A)\longrightarrow \emptyset,
\]
called a \emph{licensed discard}, such that:

\begin{enumerate}
\item for every morphism $f\colon A\to B$ in $\Ctx$,
\[
\omega_{z,B}\circ J_z(f)=\omega_{z,A};
\]

\item whenever $z\in\W_{\Lic}$ and $z\preceq z'$, the diagram
\[
\begin{tikzcd}[column sep=large,row sep=large]
J_z(A)
  \arrow[r,"\chi_{z,z',A}"]
  \arrow[dr,"\omega_{z,A}"']
&
J_{z'}(A)
  \arrow[d,"\omega_{z',A}"]
\\
& \emptyset
\end{tikzcd}
\]
commutes for every $A\in\Ob(\Ctx)$.
\end{enumerate}

\item For every $z\in\C_{\Lic}$ and every $A\in\Ob(\Ctx)$, a morphism
\[
\delta_{z,A}\colon J_z(A)\longrightarrow J_z(A)\tensor J_z(A),
\]
called a \emph{licensed diagonal}, such that:

\begin{enumerate}
\item for every morphism $f\colon A\to B$ in $\Ctx$,
\[
\bigl(J_z(f)\tensor J_z(f)\bigr)\circ\delta_{z,A}
=
\delta_{z,B}\circ J_z(f);
\]

\item whenever $z\in\C_{\Lic}$ and $z\preceq z'$, the diagram
\[
\begin{tikzcd}[column sep=large,row sep=large]
J_z(A)
  \arrow[r,"\chi_{z,z',A}"]
  \arrow[d,"\delta_{z,A}"']
&
J_{z'}(A)
  \arrow[d,"\delta_{z',A}"]
\\
J_z(A)\tensor J_z(A)
  \arrow[r,"\chi_{z,z',A}\tensor\chi_{z,z',A}"']
&
J_{z'}(A)\tensor J_{z'}(A)
\end{tikzcd}
\]
commutes for every $A\in\Ob(\Ctx)$.
\end{enumerate}
\end{enumerate}
\end{defi}

\begin{rem}\label{rem:discipline-role}
The coherence conditions say that structural behaviour is stable under refinement of zones. 
If a zone $z$ admits discard or duplication and $z\preceq z'$, then the corresponding structural operation at $z'$ is compatible with the coercion from $z$ to $z'$. 
This is the categorical precursor of the upward closure conditions in a subexponential signature.
\end{rem}

The preceding remark suggests that the structural information carried by a coherent discipline should be read not merely from the names of the designated subsets $\W_{\Lic}$ and $\C_{\Lic}$, but from the actual family of structural morphisms licensed by the architecture. We therefore reconstruct the weakening and contraction zones directly from \autoref{def:licensed-family}.
\begin{defi}[Licensed family]\label{def:licensed-family}
Let $\Lic$ be a coherent zone discipline on $\Arch_Z$. The licensed family $\Str(\Lic)$ is the collection of all structural morphisms permitted by $\Lic$.
\end{defi}

\begin{defi}[Structural classes reconstructed from the licensed family]\label{def:reconstructed-classes}
Let $\Lic$ be a coherent zone discipline on $\Arch_Z$. Define
\[
\W^{\disc}_{\Lic}
:=
\left\{
z\in Z \;\middle|\;
\forall A\in\Ob(\Ctx),\ \Str(\Lic)\text{ contains a licensed discard }
\omega_{z,A}\colon J_z(A)\to \emptyset
\right\},
\]
and
\[
\C^{\diag}_{\Lic}
:=
\left\{
z\in Z \;\middle|\;
\forall A\in\Ob(\Ctx),\ \Str(\Lic)\text{ contains a licensed diagonal }
\delta_{z,A}\colon J_z(A)\to J_z(A)\tensor J_z(A)
\right\}.
\]
\end{defi}

\begin{defi}[Extracted subexponential signature]\label{def:extracted-signature}
Let $(\Arch_Z,\Lic)$ be a coherently disciplined architectural category. Its \emph{extracted subexponential signature} is the quadruple
\[
\Sig_{\Lic}:=(Z,\preceq,\W^{\disc}_{\Lic},\C^{\diag}_{\Lic}).
\]
\end{defi}

The next theorem shows that the reconstructed classes coincide exactly with the originally designated weakening and contraction zones. In this sense, the extracted signature is not an arbitrary repackaging of the discipline, but a faithful reading of the structural permissions carried by the licensed family.

\begin{thm}[Faithful reconstruction of the structural signature]\label{thm:recovery}
Let $(\Arch_Z,\Lic)$ be a coherently disciplined architectural category. Then
\[
\W^{\disc}_{\Lic}=\W_{\Lic}
\qquad\text{and}\qquad
\C^{\diag}_{\Lic}=\C_{\Lic}.
\]
Consequently,
\[
\Sig_{\Lic}=(Z,\preceq,\W_{\Lic},\C_{\Lic}),
\]
and in particular $\Sig_{\Lic}$ is a subexponential signature.
\end{thm}

\begin{proof}
We first prove that
\[
\W^{\disc}_{\Lic}=\W_{\Lic}.
\]

Assume that
\[
z\in \W^{\disc}_{\Lic}.
\]
By \autoref{def:reconstructed-classes}, this means that for every object $A\in\Ob(\Ctx)$, the licensed structural family $\Str(\Lic)$ contains a morphism
\[
\omega_{z,A}\colon J_z(A)\to\emptyset.
\]
By \autoref{def:licensed-family}, the family $\Str(\Lic)$ is the union
\[
\Str(\Lic)=
\left\{\omega_{u,B}\;\middle|\; u\in\W_{\Lic},\ B\in\Ob(\Ctx)\right\}
\;\cup\;
\left\{\delta_{u,B}\;\middle|\; u\in\C_{\Lic},\ B\in\Ob(\Ctx)\right\}.
\]
Since the morphism
\[
\omega_{z,A}\colon J_z(A)\to\emptyset
\]
has codomain $\emptyset$, it cannot be one of the diagonals
\[
\delta_{u,B}\colon J_u(B)\to J_u(B)\tensor J_u(B),
\]
because those morphisms have codomain of the form
\[
J_u(B)\tensor J_u(B),
\]
which is not equal to $\emptyset$. Therefore
\[
\omega_{z,A}\in
\left\{\omega_{u,B}\;\middle|\; u\in\W_{\Lic},\ B\in\Ob(\Ctx)\right\}.
\]
Hence there exist $u\in\W_{\Lic}$ and $B\in\Ob(\Ctx)$ such that
\[
\omega_{z,A}=\omega_{u,B}
\]
as morphisms in $\Arch_Z$.

Equality of morphisms in a category implies equality of their domains and codomains. Thus
\[
J_z(A)=J_u(B)
\qquad\text{and}\qquad
\emptyset=\emptyset.
\]
By \autoref{def:one-zone-embedding}, one has
\[
J_z(A)=\ctx{z,A},
\qquad
J_u(B)=\ctx{u,B}.
\]
Therefore
\[
\ctx{z,A}=\ctx{u,B}.
\]
Since one-variable contexts are equal if and only if both their zone labels and their carrier objects are equal, it follows that
\[
z=u
\qquad\text{and}\qquad
A=B.
\]
In particular,
\[
z=u\in\W_{\Lic}.
\]
This proves that
\[
\W^{\disc}_{\Lic}\subseteq \W_{\Lic}.
\]

Conversely, assume that
\[
z\in\W_{\Lic}.
\]
We must prove that
\[
z\in \W^{\disc}_{\Lic}.
\]
Let $A\in\Ob(\Ctx)$ be arbitrary. By \autoref{def:discipline}, since $z\in\W_{\Lic}$, the data of the discipline include a licensed discard
\[
\omega_{z,A}\colon J_z(A)\to\emptyset.
\]
By \autoref{def:licensed-family}, every such licensed discard belongs to $\Str(\Lic)$. Since $A$ was arbitrary, we conclude that for every $A\in\Ob(\Ctx)$, the family $\Str(\Lic)$ contains a licensed discard
\[
\omega_{z,A}\colon J_z(A)\to\emptyset.
\]
Therefore, by \autoref{def:reconstructed-classes},
\[
z\in \W^{\disc}_{\Lic}.
\]
Hence
\[
\W_{\Lic}\subseteq \W^{\disc}_{\Lic}.
\]

We have proved both inclusions, so
\[
\W^{\disc}_{\Lic}=\W_{\Lic}.
\]

We now prove that
\[
\C^{\diag}_{\Lic}=\C_{\Lic}.
\]

Assume that
\[
z\in \C^{\diag}_{\Lic}.
\]
By \autoref{def:reconstructed-classes}, this means that for every object $A\in\Ob(\Ctx)$, the family $\Str(\Lic)$ contains a licensed diagonal
\[
\delta_{z,A}\colon J_z(A)\to J_z(A)\tensor J_z(A).
\]
Again by \autoref{def:licensed-family},
\[
\Str(\Lic)=
\left\{\omega_{u,B}\;\middle|\; u\in\W_{\Lic},\ B\in\Ob(\Ctx)\right\}
\;\cup\;
\left\{\delta_{u,B}\;\middle|\; u\in\C_{\Lic},\ B\in\Ob(\Ctx)\right\}.
\]
Since the morphism
\[
\delta_{z,A}\colon J_z(A)\to J_z(A)\tensor J_z(A)
\]
has codomain of the form
\[
J_z(A)\tensor J_z(A),
\]
it cannot be one of the discards
\[
\omega_{u,B}\colon J_u(B)\to\emptyset,
\]
because those morphisms have codomain $\emptyset$. Therefore
\[
\delta_{z,A}\in
\left\{\delta_{u,B}\;\middle|\; u\in\C_{\Lic},\ B\in\Ob(\Ctx)\right\}.
\]
Hence there exist $u\in\C_{\Lic}$ and $B\in\Ob(\Ctx)$ such that
\[
\delta_{z,A}=\delta_{u,B}
\]
as morphisms in $\Arch_Z$.

Equality of morphisms implies equality of domains and codomains. Therefore
\[
J_z(A)=J_u(B)
\qquad\text{and}\qquad
J_z(A)\tensor J_z(A)=J_u(B)\tensor J_u(B).
\]
From the equality of domains we already obtain
\[
J_z(A)=J_u(B).
\]
Using again \autoref{def:one-zone-embedding},
\[
J_z(A)=\ctx{z,A},
\qquad
J_u(B)=\ctx{u,B},
\]
hence
\[
\ctx{z,A}=\ctx{u,B}.
\]
Therefore
\[
z=u
\qquad\text{and}\qquad
A=B.
\]
In particular,
\[
z=u\in\C_{\Lic}.
\]
This proves that
\[
\C^{\diag}_{\Lic}\subseteq \C_{\Lic}.
\]

Conversely, assume that
\[
z\in\C_{\Lic}.
\]
We must prove that
\[
z\in \C^{\diag}_{\Lic}.
\]
Let $A\in\Ob(\Ctx)$ be arbitrary. By \autoref{def:discipline}, since $z\in\C_{\Lic}$, the data of the discipline include a licensed diagonal
\[
\delta_{z,A}\colon J_z(A)\to J_z(A)\tensor J_z(A).
\]
By \autoref{def:licensed-family}, every such licensed diagonal belongs to $\Str(\Lic)$. Since $A$ was arbitrary, it follows that for every $A\in\Ob(\Ctx)$, the family $\Str(\Lic)$ contains a licensed diagonal
\[
\delta_{z,A}\colon J_z(A)\to J_z(A)\tensor J_z(A).
\]
Hence, by \autoref{def:reconstructed-classes},
\[
z\in \C^{\diag}_{\Lic}.
\]
Therefore
\[
\C_{\Lic}\subseteq \C^{\diag}_{\Lic}.
\]

We have proved both inclusions, so
\[
\C^{\diag}_{\Lic}=\C_{\Lic}.
\]

It remains to justify the final claim. Since $\preceq$ is a preorder on $Z$ by \autoref{def:discipline}, and since
\[
\W^{\disc}_{\Lic}=\W_{\Lic},
\qquad
\C^{\diag}_{\Lic}=\C_{\Lic},
\]
the sets $\W^{\disc}_{\Lic}$ and $\C^{\diag}_{\Lic}$ are upward closed because $\W_{\Lic}$ and $\C_{\Lic}$ are upward closed by \autoref{def:discipline}. Therefore the quadruple
\[
\Sig_{\Lic}=(Z,\preceq,\W^{\disc}_{\Lic},\C^{\diag}_{\Lic})
\]
is a subexponential signature. Using the equalities already established, this is the same as
\[
\Sig_{\Lic}=(Z,\preceq,\W_{\Lic},\C_{\Lic}),
\]
as claimed.\qedhere
\end{proof}

Thus the signature is not imposed on the architecture from outside: it is reconstructed from
the structurally licensed families already present in the disciplined category.
\begin{cor}\label{cor:no-loss}
Let $(\Arch_Z,\Lic)$ be a coherently disciplined architectural category, and let
\[
\Sig_{\Lic}=(Z,\preceq,\W^{\disc}_{\Lic},\C^{\diag}_{\Lic})
\]
be the extracted signature of \autoref{def:extracted-signature}. Then the weakening and contraction zones of
\[
\Sig_{\Lic}
\]
are exactly the zones in which the discipline $\Lic$ designates discard and diagonal families.
\end{cor}

\begin{proof}
By \autoref{def:extracted-signature}, the weakening and contraction zones of $\Sig_{\Lic}$ are
\[
\W^{\disc}_{\Lic}
\qquad\text{and}\qquad
\C^{\diag}_{\Lic},
\]
respectively. By \autoref{thm:recovery},
\[
\W^{\disc}_{\Lic}=\W_{\Lic}
\qquad\text{and}\qquad
\C^{\diag}_{\Lic}=\C_{\Lic}.
\]
By definition, the sets $\W_{\Lic}$ and $\C_{\Lic}$ are exactly the zones in which the discipline
$\Lic$ supplies licensed discard and diagonal families. Therefore the extracted signature retains
exactly the structural zone information carried by the discipline.\qedhere
\end{proof}
The previous results complete the passage from the architectural category to a subexponential signature. 
What remains is to show how this extracted signature governs the proof theory. 
That step belongs to the next section, where the tensorial calculus $\TZ_{\Sig}$ is defined and the soundness of its structural rules is established with respect to the disciplined architecture.

\section{The extracted proof theory}\label{sec:proof}

The previous section showed how a coherent zone discipline on the architectural category determines a subexponential signature
\[
\Sig_{\Lic}=(Z,\preceq,\W^{\disc}_{\Lic},\C^{\diag}_{\Lic}).
\]
By \autoref{thm:recovery}, one may equivalently write
\[
\Sig_{\Lic}=(Z,\preceq,\W_{\Lic},\C_{\Lic}).
\]
We now pass from this extracted signature to a sequent calculus. 
This is the syntactic step in the chain
\[
\text{architecture}\;\longrightarrow\;\text{categorisation}\;\longrightarrow\;\text{subexponential signature}\;\longrightarrow\;\text{proof theory}.
\]
The guiding principle is that the proof system should remember exactly the structural permissions extracted from the architecture and nothing more.

The fragment considered here is intentionally modest. 
It contains atoms, the tensor unit, tensor product, and the zone-dependent structural rules of weakening and contraction. 
This is the part of the logic directly forced by the architectural data developed so far. 
In particular, no implication is introduced at this stage, since its correct treatment would require additional categorical structure not yet available in the present framework.The tensorial system considered here should be read as a deliberately restricted,
zone-indexed fragment of intuitionistic linear proof theory; for a classical term-calculus
reference point on intuitionistic linear logic, see \cite{BBPH93}.

\begin{defi}[Subexponential signature]\label{def:signature}
A \emph{subexponential signature} is a quadruple
\[
\Sig=(Z,\preceq,\W,\C),
\]
where $Z$ is a set of zones, $\preceq$ is a preorder on $Z$, and $\W,\C\subseteq Z$ are upward-closed subsets.
The elements of $\W$ are the zones in which weakening is allowed, and the elements of $\C$ are the zones in which contraction is allowed.
\end{defi}

Although the signature of interest for this article is the extracted signature $\Sig_{\Lic}$ of \autoref{def:extracted-signature}, the proof system is most conveniently stated for an arbitrary signature.

\begin{defi}[Tensorial formulas and zone contexts]\label{def:formulas}
Fix a set $\mathsf{At}$ of atomic formulas. 
The set of \emph{tensorial formulas} is inductively generated by
\[
A ::= p \mid \unit \mid A\tensor B
\qquad (p\in\mathsf{At}).
\]

A \emph{zone context} is a finite multiset of expressions
\[
\zone{z}{A},
\]
where $z\in Z$ and $A$ is a tensorial formula. 
If $\Gamma$ and $\Delta$ are zone contexts, their concatenation is written $\Gamma,\Delta$. 
Contexts are taken modulo permutation.
\end{defi}

The preorder $\preceq$ is retained because it is part of the subexponential data extracted from the architecture and will be relevant for richer fragments. 
In the tensorial system defined below, however, only the sets $\W$ and $\C$ occur explicitly.

\begin{defi}[The extracted zone calculus $\TZ_{\Sig}$]\label{def:tz}
The calculus $\TZ_{\Sig}$ has judgments of the form
\[
\Gamma\vdash A,
\]
where $\Gamma$ is a zone context and $A$ is a tensorial formula. 
Its inference rules are the following:

\[
\frac{}{\zone{z}{A}\vdash A}\;\mathrm{ax}
\qquad
\frac{}{\emptyset\vdash \unit}\;\unit\mathrm{R}
\qquad
\frac{\Gamma\vdash A \qquad \Delta\vdash B}{\Gamma,\Delta\vdash A\tensor B}\;\tensor\mathrm{R}
\]

\[
\frac{\Gamma,\zone{z}{A},\zone{z}{B}\vdash C}{\Gamma,\zone{z}{A\tensor B}\vdash C}\;\tensor\mathrm{L}
\qquad
\frac{\Gamma\vdash C}{\Gamma,\zone{z}{\unit}\vdash C}\;\unit\mathrm{L}
\]

\[
\frac{\Gamma\vdash A \qquad \Delta,\zone{z}{A}\vdash B}{\Gamma,\Delta\vdash B}\;\mathrm{cut}
\]

For each zone $z\in\W$, there is a weakening rule
\[
\frac{\Gamma\vdash B}{\Gamma,\zone{z}{A}\vdash B}\;\mathrm{W}_z,
\]
and for each zone $z\in\C$, there is a contraction rule
\[
\frac{\Gamma,\zone{z}{A},\zone{z}{A}\vdash B}{\Gamma,\zone{z}{A}\vdash B}\;\mathrm{C}_z.
\]
\end{defi}

\begin{rem}\label{rem:scope}
The choice of fragment is deliberate. 
It isolates the part of the logic that is directly forced by the structural architectural behaviour encoded in the signature: tensorial aggregation, the tensor unit, and zone-dependent weakening and contraction. 
A satisfactory extension to linear implication would require additional semantic structure, such as an internal hom compatible with the resource discipline or a linear--nonlinear adjunction; see \cite{Benton1995,Rogozin2025}. 
The present article does not assume such structure.
\end{rem}

The next proposition makes precise that the extracted proof theory depends only on the extracted signature, and not on further accidental features of the disciplined architecture.

\begin{thm}[Invariance of the extracted logic under structural equivalence]\label{thm:invariance}
Let $\Lic$ and $\Lic'$ be coherent zone disciplines on the same architectural category $\Arch_Z$.
Assume that they induce the same preorder on $Z$ and have the same licensed structural family:
\[
\preceq_{\Lic}=\preceq_{\Lic'}
\qquad\text{and}\qquad
\Str(\Lic)=\Str(\Lic').
\]
Then
\[
\W^{\disc}_{\Lic}=\W^{\disc}_{\Lic'}
\qquad\text{and}\qquad
\C^{\diag}_{\Lic}=\C^{\diag}_{\Lic'}.
\]
Consequently,
\[
\Sig_{\Lic}=\Sig_{\Lic'},
\]
and therefore
\[
\TZ_{\Sig_{\Lic}}=\TZ_{\Sig_{\Lic'}}.
\]
\end{thm}

\begin{proof}
We first prove that
\[
\W^{\disc}_{\Lic}=\W^{\disc}_{\Lic'}.
\]

Assume that
\[
z\in \W^{\disc}_{\Lic}.
\]
By \autoref{def:reconstructed-classes}, this means that for every object $A\in\Ob(\Ctx)$, the family
\[
\Str(\Lic)
\]
contains a licensed discard
\[
\omega_{z,A}\colon J_z(A)\to \emptyset.
\]
Since
\[
\Str(\Lic)=\Str(\Lic'),
\]
the same morphism belongs to $\Str(\Lic')$. As this holds for every $A\in\Ob(\Ctx)$, the definition of $\W^{\disc}_{\Lic'}$ yields
\[
z\in \W^{\disc}_{\Lic'}.
\]
Hence
\[
\W^{\disc}_{\Lic}\subseteq \W^{\disc}_{\Lic'}.
\]

Conversely, assume that
\[
z\in \W^{\disc}_{\Lic'}.
\]
Then for every $A\in\Ob(\Ctx)$, the family
\[
\Str(\Lic')
\]
contains a licensed discard
\[
\omega_{z,A}\colon J_z(A)\to \emptyset.
\]
Again using
\[
\Str(\Lic')=\Str(\Lic),
\]
we conclude that for every $A\in\Ob(\Ctx)$, the family $\Str(\Lic)$ contains such a licensed discard. Therefore
\[
z\in \W^{\disc}_{\Lic}.
\]
Thus
\[
\W^{\disc}_{\Lic'}\subseteq \W^{\disc}_{\Lic}.
\]

We have proved both inclusions, so
\[
\W^{\disc}_{\Lic}=\W^{\disc}_{\Lic'}.
\]

We next prove that
\[
\C^{\diag}_{\Lic}=\C^{\diag}_{\Lic'}.
\]

Assume that
\[
z\in \C^{\diag}_{\Lic}.
\]
By \autoref{def:reconstructed-classes}, this means that for every object $A\in\Ob(\Ctx)$, the family
\[
\Str(\Lic)
\]
contains a licensed diagonal
\[
\delta_{z,A}\colon J_z(A)\to J_z(A)\tensor J_z(A).
\]
Since
\[
\Str(\Lic)=\Str(\Lic'),
\]
the same morphism belongs to $\Str(\Lic')$. As this holds for every $A\in\Ob(\Ctx)$, we obtain
\[
z\in \C^{\diag}_{\Lic'}.
\]
Hence
\[
\C^{\diag}_{\Lic}\subseteq \C^{\diag}_{\Lic'}.
\]

Conversely, assume that
\[
z\in \C^{\diag}_{\Lic'}.
\]
Then for every $A\in\Ob(\Ctx)$, the family
\[
\Str(\Lic')
\]
contains a licensed diagonal
\[
\delta_{z,A}\colon J_z(A)\to J_z(A)\tensor J_z(A).
\]
Using again
\[
\Str(\Lic')=\Str(\Lic),
\]
we conclude that for every $A\in\Ob(\Ctx)$, the family $\Str(\Lic)$ contains such a licensed diagonal. Therefore
\[
z\in \C^{\diag}_{\Lic}.
\]
Thus
\[
\C^{\diag}_{\Lic'}\subseteq \C^{\diag}_{\Lic}.
\]

We have proved both inclusions, and therefore
\[
\C^{\diag}_{\Lic}=\C^{\diag}_{\Lic'}.
\]

By \autoref{def:extracted-signature}, one has
\[
\Sig_{\Lic}=(Z,\preceq_{\Lic},\W^{\disc}_{\Lic},\C^{\diag}_{\Lic})
\]
and
\[
\Sig_{\Lic'}=(Z,\preceq_{\Lic'},\W^{\disc}_{\Lic'},\C^{\diag}_{\Lic'}).
\]
Since the preorders agree by assumption and the reconstructed classes have just been shown to agree, it follows that
\[
\Sig_{\Lic}=\Sig_{\Lic'}.
\]

It remains to prove that
\[
\TZ_{\Sig_{\Lic}}=\TZ_{\Sig_{\Lic'}}.
\]
By \autoref{def:tz}, the rule schemes of $\TZ_{\Sig}$ are determined by the following data:
\begin{enumerate}
\item the zone set $Z$, which determines the labels that may occur in zone contexts;
\item the weakening class $\W$, which determines exactly for which zones the rules $\mathrm{W}_z$ are present;
\item the contraction class $\C$, which determines exactly for which zones the rules $\mathrm{C}_z$ are present.
\end{enumerate}
The remaining rules
\[
\mathrm{ax},\quad \unit\mathrm{R},\quad \tensor\mathrm{R},\quad \tensor\mathrm{L},\quad \unit\mathrm{L},\quad \mathrm{cut}
\]
are independent of any further data. Since
\[
\Sig_{\Lic}=\Sig_{\Lic'},
\]
the two calculi have the same zone set and the same weakening and contraction classes. Therefore they have exactly the same inference rules, and hence
\[
\TZ_{\Sig_{\Lic}}=\TZ_{\Sig_{\Lic'}}.
\]
This completes the proof.\qedhere
\end{proof}

The preceding theorem formalises a central point of the paper: the extracted proof theory is invariant under replacement of a coherent discipline by another one carrying the same structural family and the same zone preorder. Thus the calculus depends only on the reconstructed structural data, and not on further accidental features of the chosen disciplined presentation.

We now turn to the main proof-theoretic property of the system, namely admissibility of cut. 
Since $\TZ_{\Sig}$ is a tensorial fragment of intuitionistic linear logic with zone-dependent structural rules, it is natural to derive cut elimination from the known cut-elimination theorems for subexponential systems. 
To do this rigorously, one must still verify that the cut-free derivation produced in the ambient system remains inside the tensorial fragment.

The technical bridge between the ambient subexponential metatheory and the extracted
tensorial calculus is the following closure lemma, which ensures that cut elimination in the
ambient system descends to the extracted fragment itself.

\begin{lem}[Closure of the tensorial fragment under cut-free derivations]\label{lem:fragment-closure}
Let $\Pi$ be a cut-free derivation in a standard single-conclusion sequent calculus for intuitionistic linear logic with subexponentials, with signature $\Sig=(Z,\preceq,\W,\C)$. 
Assume that the end-sequent of $\Pi$ has the form
\[
\Gamma\vdash A,
\]
where every formula occurring in $\Gamma$ and $A$ belongs to the tensorial grammar of \autoref{def:formulas}. 
Then every formula occurring anywhere in $\Pi$ is tensorial, and every inference of $\Pi$ is an instance of one of the rules of $\TZ_{\Sig}$.
\end{lem}

\begin{proof}
We argue by induction on the height of the cut-free derivation $\Pi$.

If $\Pi$ has height \(0\), then $\Pi$ consists of a single axiom. 
In a standard single-conclusion sequent calculus, the only cut-free derivation of height \(0\) with a nonempty succedent is an identity axiom. 
Its conclusion is therefore of the form
\[
\zone{z}{B}\vdash B
\]
for some tensorial formula $B$. 
This is an instance of the rule $\mathrm{ax}$ of $\TZ_{\Sig}$, and every formula occurring in the derivation is tensorial.

Assume now that $\Pi$ has positive height, and let its last inference be \(R\). 
We examine the possible forms of \(R\).

\medskip

\noindent\emph{Case 1: \(R\) is an identity axiom.}
This is the same as the height-\(0\) case and is already covered.

\medskip

\noindent\emph{Case 2: \(R\) is the right rule for the tensor unit.}
Then the conclusion of \(R\) is
\[
\emptyset\vdash \unit.
\]
This is an instance of \(\unit\mathrm{R}\), and there are no premises. 
Hence the derivation lies in \(\TZ_{\Sig}\).

\medskip

\noindent\emph{Case 3: \(R\) is the right rule for tensor.}
Then the conclusion of \(R\) is
\[
\Gamma,\Delta\vdash A\tensor B
\]
for some tensorial formulas \(A\) and \(B\), and the premises are
\[
\Gamma\vdash A
\qquad\text{and}\qquad
\Delta\vdash B.
\]
Since \(\Pi\) is cut-free, both premises are derivable by cut-free subderivations of strictly smaller height. 
By the induction hypothesis, every formula in those subderivations is tensorial and every inference in them belongs to \(\TZ_{\Sig}\). 
The last rule \(R\) is precisely \(\tensor\mathrm{R}\). 
Therefore the whole derivation belongs to \(\TZ_{\Sig}\).

\medskip

\noindent\emph{Case 4: \(R\) is the left rule for tensor.}
Then the conclusion of \(R\) is
\[
\Gamma,\zone{z}{A\tensor B}\vdash C,
\]
and its premise is
\[
\Gamma,\zone{z}{A},\zone{z}{B}\vdash C.
\]
Since \(A\tensor B\) is tensorial, both \(A\) and \(B\) are tensorial. 
The premise is therefore a sequent all of whose formulas are tensorial. 
By the induction hypothesis, the subderivation of the premise belongs to \(\TZ_{\Sig}\). 
The last rule \(R\) is exactly \(\tensor\mathrm{L}\). 
Hence the whole derivation belongs to \(\TZ_{\Sig}\).

\medskip

\noindent\emph{Case 5: \(R\) is the left rule for the tensor unit.}
Then the conclusion is
\[
\Gamma,\zone{z}{\unit}\vdash C,
\]
and the premise is
\[
\Gamma\vdash C.
\]
Again the premise contains only tensorial formulas, so by the induction hypothesis its subderivation lies in \(\TZ_{\Sig}\). 
The last rule is precisely \(\unit\mathrm{L}\). 
Therefore the entire derivation lies in \(\TZ_{\Sig}\).

\medskip

\noindent\emph{Case 6: \(R\) is a weakening rule.}
In a standard subexponential sequent calculus, such a rule may be applied only in zones that admit weakening. 
Its conclusion has the form
\[
\Gamma,\zone{z}{A}\vdash B,
\]
and its premise has the form
\[
\Gamma\vdash B.
\]
Since the conclusion contains only tensorial formulas, the premise does as well. 
By the induction hypothesis, the subderivation of the premise belongs to \(\TZ_{\Sig}\). 
Since \(z\in\W\), the rule \(R\) is exactly an instance of \(\mathrm{W}_z\). 
Thus the entire derivation lies in \(\TZ_{\Sig}\).

\medskip

\noindent\emph{Case 7: \(R\) is a contraction rule.}
Its conclusion has the form
\[
\Gamma,\zone{z}{A}\vdash B,
\]
and its premise has the form
\[
\Gamma,\zone{z}{A},\zone{z}{A}\vdash B.
\]
Again all formulas occurring are tensorial, so by the induction hypothesis the subderivation of the premise belongs to \(\TZ_{\Sig}\). 
Since \(z\in\C\), the rule \(R\) is exactly an instance of \(\mathrm{C}_z\). 
Hence the whole derivation lies in \(\TZ_{\Sig}\).

\medskip

\noindent\emph{Case 8: \(R\) is exchange.}
Since contexts in \(\TZ_{\Sig}\) are taken modulo permutation, exchange is implicit in our presentation. 
Thus an exchange step does not take the derivation outside the fragment. 
Its premises contain exactly the same formulas as its conclusion, merely in different order, so the induction hypothesis applies directly.

\medskip

\noindent\emph{Case 9: \(R\) is any other rule of the ambient subexponential calculus.}
We show that this is impossible.

Any right rule for a connective other than \(\unit\) or \(\tensor\) would have a conclusion whose succedent is headed by that connective. 
Since the succedent formula of the conclusion of \(R\) is tensorial, no such rule can occur.

Any left rule for a connective or modality other than \(\unit\) or \(\tensor\) would have a conclusion containing, in its antecedent, a principal formula headed by that connective or modality. 
Since every antecedent formula in the conclusion is tensorial, no such rule can occur either.

Any modal rule introducing or manipulating explicit subexponential formulas is likewise impossible, because no formula in the conclusion carries such a connective.

Finally, the derivation is cut-free by assumption, so \(R\) cannot be a cut.

Thus no other case can occur.

\medskip

Having exhausted all possibilities for the last rule \(R\), we conclude that every formula in \(\Pi\) is tensorial and every inference of \(\Pi\) belongs to \(\TZ_{\Sig}\). This completes the induction.\qedhere
\end{proof}

We may now state the main proof-theoretic theorem for the extracted calculus.

\begin{thm}[Cut elimination]\label{thm:cut}
For every subexponential signature \(\Sig\), the cut rule is admissible in the calculus \(\TZ_{\Sig}\). 
Equivalently, every derivable sequent of \(\TZ_{\Sig}\) has a cut-free derivation in \(\TZ_{\Sig}\).
\end{thm}

\begin{proof}
Let
\[
\Pi
\]
be a derivation in \(\TZ_{\Sig}\) of some sequent
\[
\Gamma\vdash A.
\]
By construction, every rule of \(\TZ_{\Sig}\) is an instance of a rule of the standard single-conclusion sequent calculi for intuitionistic linear logic with subexponentials studied in the literature; see, for example, \cite{NigamMiller2009,XavierOlartePimentel2022}. 
Therefore \(\Pi\) may be viewed, without any modification, as a derivation in such an ambient subexponential calculus.

By the cut-elimination theorem for that ambient calculus, there exists a cut-free derivation
\[
\Pi'
\]
of the same end-sequent
\[
\Gamma\vdash A.
\]
Since the end-sequent contains only tensorial formulas, \autoref{lem:fragment-closure} applies to \(\Pi'\). 
It follows that every formula occurring in \(\Pi'\) is tensorial and every inference of \(\Pi'\) belongs to the rule set of \(\TZ_{\Sig}\).

Hence \(\Pi'\) is in fact a cut-free derivation in \(\TZ_{\Sig}\) itself. 
Therefore the cut rule is admissible in \(\TZ_{\Sig}\), and every derivable sequent has a cut-free derivation in the same calculus.\qedhere
\end{proof}

Although the general cut-elimination theorem is imported from the known subexponential literature, it is useful to record the principal reduction specific to the tensorial fragment, since it is the only genuinely connective reduction used in the system.

\begin{rem}[Principal cut on tensor]\label{rem:principal-cut}
Suppose one has derivations
\[
\frac{\Pi_1\colon \Gamma\vdash A \qquad \Pi_2\colon \Delta\vdash B}
     {\Gamma,\Delta\vdash A\tensor B}\;\tensor\mathrm{R}
\]
and
\[
\frac{\Pi_3\colon \Theta,\zone{z}{A},\zone{z}{B}\vdash C}
     {\Theta,\zone{z}{A\tensor B}\vdash C}\;\tensor\mathrm{L}.
\]
A cut on the principal formula \(A\tensor B\) reduces to two successive cuts:
\begin{center}
{\large
\begin{prooftree}
\AxiomC{$\Pi_2\colon \Delta\vdash B$}
\AxiomC{$\Pi_1\colon \Gamma\vdash A$}
\AxiomC{$\Pi_3\colon \Theta,\zone{z}{A},\zone{z}{B}\vdash C$}
\RightLabel{\scriptsize cut}
\BinaryInfC{$\Gamma,\Theta,\zone{z}{B}\vdash C$}
\RightLabel{\scriptsize cut}
\BinaryInfC{$\Gamma,\Delta,\Theta\vdash C$}
\end{prooftree}
}
\end{center}
up to permutation of the multiset context.
The two new cuts are on the strictly smaller formulas \(A\) and \(B\). 
This is the tensorial core of the cut-reduction argument carried out in the general systems cited above.
\end{rem}

\begin{rem}\label{rem:focusing}
The proof-theoretic stance adopted here is compatible with the broader literature on focusing and subexponentials; see \cite{Andreoli1992,NigamMiller2009,XavierOlartePimentel2022}. 
We do not develop a focused presentation in the present article. 
For the purposes of this work, \autoref{thm:cut} is the essential proof-theoretic property needed to support the extracted tensorial system.
\end{rem}

\section{Soundness in the categorized architecture}\label{sec:soundness}

Let $(\Arch_Z,\Lic)$ be a coherently disciplined architectural category, and let
\[
\Sig_{\Lic}=(Z,\preceq,\W^{\disc}_{\Lic},\C^{\diag}_{\Lic})
\]
be its extracted subexponential signature. By \autoref{thm:recovery}, one may equivalently write
\[
\Sig_{\Lic}=(Z,\preceq,\W_{\Lic},\C_{\Lic}).
\]
We now return from syntax to architecture and show that every derivation in the extracted calculus is interpreted by a morphism built from the architectural data and the designated structural families.

The interpretation of sequents requires a harmless extension of the label set. 
Indeed, while antecedent formulas are interpreted in the zones prescribed by the context, the succedent must be represented uniformly as an output object. 
To achieve this, we adjoin a fresh output label.

\begin{defi}[Output extension of the zone set]\label{def:output-extension}
Fix a symbol $\out\notin Z$, and let
\[
\widehat Z:=Z\sqcup\{\out\}.
\]
All categorical constructions of Section~\ref{sec:arch} may be carried out with $\widehat Z$ in place of $Z$. 
We write
\[
\Arch_{\widehat Z}
\]
for the resulting architectural category. 
By construction, every object and morphism of $\Arch_Z$ may be regarded canonically as an object or morphism of $\Arch_{\widehat Z}$.
\end{defi}

Thus the disciplined architecture embeds canonically into a larger category in which a dedicated output zone is available. 
Only this output zone is new; no additional structural permissions are attached to it.

We now interpret formulas and contexts.

\begin{defi}[Formula and context interpretation]\label{def:formula-semantics}
Fix an assignment
\[
\rho\colon \mathsf{At}\to \Ob(\Ctx).
\]
Extend it recursively to tensorial formulas by
\[
\sem{p}:=\rho(p),
\qquad
\sem{\unit}:=1,
\qquad
\sem{A\tensor B}:=\sem{A}\times\sem{B}.
\]

For each zone $z\in \widehat Z$, let
\[
J_z\colon \Ctx\to \Arch_{\widehat Z}
\]
be the one-zone embedding of \autoref{def:one-zone-embedding}.

If
\[
\Gamma=(\zone{z_1}{A_1},\dots,\zone{z_n}{A_n})
\]
is an ordered representative of a zone context, define
\[
\sem{\Gamma}:=
J_{z_1}(\sem{A_1})\tensor\cdots\tensor J_{z_n}(\sem{A_n}).
\]
For the empty context, set
\[
\sem{\emptyset}:=\emptyset.
\]

Finally, define the \emph{output object} of a formula $A$ by
\[
\Out(A):=J_{\out}(\sem{A}).
\]
\end{defi}

\begin{rem}\label{rem:context-ordering}
Since zone contexts are taken modulo permutation, the object $\sem{\Gamma}$ depends a priori on the chosen ordered representative. 
However, any two representatives are related by a canonical composite of associators and symmetries in the symmetric monoidal category $\Arch_{\widehat Z}$. 
Because these structural isomorphisms will belong to the diagrammatic class defined below, the existence of an interpretation morphism is independent of the chosen representative.
\end{rem}

To interpret the logical rules we require, besides the structural maps already supplied by the discipline, a small family of canonical interface morphisms connecting zone-labelled inputs to the output zone.

\begin{defi}[Canonical semantic interface blocks]\label{def:semantic-interface}
For every zone $z\in Z$ and every tensorial formula $A$, define the following morphisms in $\Arch_{\widehat Z}$.

\begin{enumerate}
\item The \emph{readout block}
\[
r_{z,A}\colon J_z(\sem{A})\longrightarrow \Out(A)
\]
is the parametrised block induced by the identity morphism
\[
\id_{\sem{A}}\colon \sem{A}\to \sem{A}.
\]

\item The \emph{input block}
\[
\iota_{z,A}\colon \Out(A)\longrightarrow J_z(\sem{A})
\]
is the parametrised block induced by the same identity morphism
\[
\id_{\sem{A}}\colon \sem{A}\to \sem{A}.
\]

\item The \emph{tensor output block}
\[
\tau_{A,B}\colon \Out(A)\tensor\Out(B)\longrightarrow \Out(A\tensor B)
\]
is the parametrised block induced by the identity on
\[
\sem{A}\times\sem{B}=\sem{A\tensor B}.
\]

\item The \emph{tensor decomposition block}
\[
\mu_{z,A,B}\colon J_z(\sem{A\tensor B})\longrightarrow J_z(\sem{A})\tensor J_z(\sem{B})
\]
is the parametrised block induced by the identity on
\[
\sem{A\tensor B}=\sem{A}\times\sem{B}.
\]

\item The \emph{unit output block}
\[
\eta\colon \emptyset\longrightarrow \Out(\unit)
\]
is the parametrised block induced by the identity on the terminal object
\[
1=\sem{\unit}.
\]

\item For each $z\in Z$, the \emph{unit discard block}
\[
\nu_z\colon J_z(\sem{\unit})\longrightarrow \emptyset
\]
is the parametrised block induced by the unique morphism
\[
1\to 1.
\]
\end{enumerate}
\end{defi}

\begin{rem}\label{rem:interface-role}
The blocks of \autoref{def:semantic-interface} are not structural permissions. 
They play a different role: they connect the zone-indexed antecedent semantics with the distinguished output zone, and they unpack the tensor and unit at the semantic level. 
The only genuinely resource-sensitive maps remain the licensed discards and diagonals supplied by the discipline.
\end{rem}

We may now define the class of diagrams through which all derivations will be interpreted.

\begin{defi}[Licensed diagram class]\label{def:licensed-diagrams}
Let $\Diag_{\Lic}$ be the smallest class of morphisms in $\Arch_{\widehat Z}$ satisfying the following conditions.

\begin{enumerate}
\item It contains all identities, associators, unitors, and symmetries of the symmetric monoidal structure.

\item It contains all interface blocks of \autoref{def:semantic-interface}.

\item It contains all coercions
\[
\chi_{z,z',X}\colon J_z(X)\to J_{z'}(X)
\]
supplied by the discipline $\Lic$, for all $z\preceq z'$ in $Z$ and all $X\in\Ob(\Ctx)$.

\item It contains all licensed discards
\[
\omega_{z,X}\colon J_z(X)\to \emptyset
\]
supplied by $\Lic$, for all $z\in\W_{\Lic}$ and all $X\in\Ob(\Ctx)$.

\item It contains all licensed diagonals
\[
\delta_{z,X}\colon J_z(X)\to J_z(X)\tensor J_z(X)
\]
supplied by $\Lic$, for all $z\in\C_{\Lic}$ and all $X\in\Ob(\Ctx)$.

\item It is closed under composition.

\item It is closed under tensor product.
\end{enumerate}
\end{defi}

\begin{lem}\label{lem:diag-stability}
The class $\Diag_{\Lic}$ is closed under precomposition and postcomposition with canonical associators, unitors, and symmetries in $\Arch_{\widehat Z}$.
\end{lem}

\begin{proof}
By \autoref{def:licensed-diagrams}(1), every associator, unitor, and symmetry belongs to $\Diag_{\Lic}$. 
By \autoref{def:licensed-diagrams}(6), the class is closed under composition. 
Hence, if
\[
f\in \Diag_{\Lic}
\]
and
\[
\sigma,\tau\in\Diag_{\Lic}
\]
are canonical structural isomorphisms such that the composites are defined, then
\[
\sigma\circ f,\qquad f\circ\tau,\qquad \sigma\circ f\circ\tau
\]
all belong to $\Diag_{\Lic}$.\qedhere
\end{proof}

We are now in position to prove the soundness theorem.

\begin{thm}[Soundness of the extracted calculus]\label{thm:soundness}
For every derivation
\[
\pi\colon \Gamma\vdash A
\]
in the extracted calculus
\[
\TZ_{\Sig_{\Lic}},
\]
there exists a morphism
\[
\sem{\pi}\colon \sem{\Gamma}\longrightarrow \Out(A)
\]
belonging to $\Diag_{\Lic}$.

Moreover, if the last rule of $\pi$ is an instance of $\mathrm{W}_z$, then the construction of $\sem{\pi}$ uses the designated discard family
\[
\omega_{z,-};
\]
if the last rule of $\pi$ is an instance of $\mathrm{C}_z$, then the construction of $\sem{\pi}$ uses the designated diagonal family
\[
\delta_{z,-}.
\]
\end{thm}

\begin{proof}
We proceed by induction on the derivation $\pi$.

\medskip

\noindent\emph{Case 1: $\pi$ ends with $\mathrm{ax}$.}
Then
\[
\pi\;=\;
\frac{}{\zone{z}{A}\vdash A}\;\mathrm{ax}.
\]
By definition,
\[
\sem{\zone{z}{A}}=J_z(\sem{A}),
\qquad
\Out(A)=J_{\out}(\sem{A}).
\]
Define
\[
\sem{\pi}:=r_{z,A}\colon J_z(\sem{A})\to \Out(A).
\]
By \autoref{def:licensed-diagrams}(2), the readout block belongs to $\Diag_{\Lic}$.

\medskip

\noindent\emph{Case 2: $\pi$ ends with $\unit\mathrm{R}$.}
Then
\[
\pi\;=\;
\frac{}{\emptyset\vdash \unit}\;\unit\mathrm{R}.
\]
Here
\[
\sem{\emptyset}=\emptyset
\qquad\text{and}\qquad
\Out(\unit)=J_{\out}(1).
\]
Define
\[
\sem{\pi}:=\eta\colon \emptyset\to \Out(\unit).
\]
By \autoref{def:licensed-diagrams}(2), one has
\[
\eta\in\Diag_{\Lic}.
\]

\medskip

\noindent\emph{Case 3: $\pi$ ends with $\tensor\mathrm{R}$.}
Then
\[
\pi\;=\;
\frac{\pi_1\colon \Gamma\vdash A \qquad \pi_2\colon \Delta\vdash B}
     {\Gamma,\Delta\vdash A\tensor B}\;\tensor\mathrm{R}.
\]
By the induction hypothesis, there exist morphisms
\[
\sem{\pi_1}\colon \sem{\Gamma}\to \Out(A),
\qquad
\sem{\pi_2}\colon \sem{\Delta}\to \Out(B)
\]
in $\Diag_{\Lic}$.
By closure under tensor,
\[
\sem{\pi_1}\tensor \sem{\pi_2}\colon
\sem{\Gamma}\tensor\sem{\Delta}\to \Out(A)\tensor\Out(B)
\]
lies in $\Diag_{\Lic}$.
By \autoref{def:licensed-diagrams}(2), the tensor output block
\[
\tau_{A,B}\colon \Out(A)\tensor\Out(B)\to \Out(A\tensor B)
\]
belongs to $\Diag_{\Lic}$.
Define
\[
\sem{\pi}:=\tau_{A,B}\circ(\sem{\pi_1}\tensor\sem{\pi_2}).
\]
Then
\[
\sem{\pi}\colon \sem{\Gamma}\tensor\sem{\Delta}\to \Out(A\tensor B)
\]
lies in $\Diag_{\Lic}$.
Since $\sem{\Gamma,\Delta}$ is, by definition, represented by $\sem{\Gamma}\tensor\sem{\Delta}$ up to canonical associativity and symmetry isomorphism, \autoref{lem:diag-stability} shows that $\sem{\pi}$ may be regarded as a morphism
\[
\sem{\Gamma,\Delta}\to \Out(A\tensor B)
\]
in $\Diag_{\Lic}$.

\medskip

\noindent\emph{Case 4: $\pi$ ends with $\tensor\mathrm{L}$.}
Then
\[
\pi\;=\;
\frac{\pi_0\colon \Gamma,\zone{z}{A},\zone{z}{B}\vdash C}
     {\Gamma,\zone{z}{A\tensor B}\vdash C}\;\tensor\mathrm{L}.
\]
By the induction hypothesis, there exists a morphism
\[
\sem{\pi_0}\colon \sem{\Gamma,\zone{z}{A},\zone{z}{B}}\to \Out(C)
\]
in $\Diag_{\Lic}$.

By definition of context semantics,
\[
\sem{\Gamma,\zone{z}{A\tensor B}}
\]
and
\[
\sem{\Gamma}\tensor J_z(\sem{A\tensor B})
\]
differ only by canonical associativity and symmetry isomorphisms. Likewise,
\[
\sem{\Gamma,\zone{z}{A},\zone{z}{B}}
\]
is canonically isomorphic to
\[
\sem{\Gamma}\tensor J_z(\sem{A})\tensor J_z(\sem{B}).
\]
By \autoref{def:licensed-diagrams}(2), the block
\[
\mu_{z,A,B}\colon J_z(\sem{A\tensor B})\to J_z(\sem{A})\tensor J_z(\sem{B})
\]
belongs to $\Diag_{\Lic}$.
Hence
\[
\id_{\sem{\Gamma}}\tensor \mu_{z,A,B}\colon
\sem{\Gamma}\tensor J_z(\sem{A\tensor B})
\to
\sem{\Gamma}\tensor J_z(\sem{A})\tensor J_z(\sem{B})
\]
belongs to $\Diag_{\Lic}$ by closure under tensor.
Composing with the relevant structural isomorphisms and then with $\sem{\pi_0}$, we obtain a morphism
\[
\sem{\pi}\colon \sem{\Gamma,\zone{z}{A\tensor B}}\to \Out(C)
\]
in $\Diag_{\Lic}$.
More explicitly, $\sem{\pi}$ is the composite
\[
\sem{\Gamma,\zone{z}{A\tensor B}}
\xrightarrow{\sigma_1}
\sem{\Gamma}\tensor J_z(\sem{A\tensor B})
\xrightarrow{\id\tensor\mu_{z,A,B}}
\sem{\Gamma}\tensor J_z(\sem{A})\tensor J_z(\sem{B})
\xrightarrow{\sigma_2}
\sem{\Gamma,\zone{z}{A},\zone{z}{B}}
\xrightarrow{\sem{\pi_0}}
\Out(C),
\]
where $\sigma_1$ and $\sigma_2$ are canonical structural isomorphisms.

\medskip

\noindent\emph{Case 5: $\pi$ ends with $\unit\mathrm{L}$.}
Then
\[
\pi\;=\;
\frac{\pi_0\colon \Gamma\vdash C}
     {\Gamma,\zone{z}{\unit}\vdash C}\;\unit\mathrm{L}.
\]
By the induction hypothesis,
\[
\sem{\pi_0}\colon \sem{\Gamma}\to \Out(C)
\]
lies in $\Diag_{\Lic}$.
Since
\[
\sem{\zone{z}{\unit}}=J_z(1),
\]
the block
\[
\nu_z\colon J_z(1)\to \emptyset
\]
belongs to $\Diag_{\Lic}$ by \autoref{def:licensed-diagrams}(2).
Using closure under tensor, we obtain
\[
\id_{\sem{\Gamma}}\tensor\nu_z\colon
\sem{\Gamma}\tensor J_z(1)\to \sem{\Gamma}\tensor \emptyset.
\]
Composing with the right unitor
\[
\rho_{\sem{\Gamma}}\colon \sem{\Gamma}\tensor \emptyset\to \sem{\Gamma}
\]
and then with $\sem{\pi_0}$ yields a morphism
\[
\sem{\Gamma}\tensor J_z(1)\to \Out(C)
\]
in $\Diag_{\Lic}$.
As in the tensor-left case, after inserting the canonical structural isomorphism between
\[
\sem{\Gamma,\zone{z}{\unit}}
\quad\text{and}\quad
\sem{\Gamma}\tensor J_z(1),
\]
we obtain the required interpretation
\[
\sem{\pi}\colon \sem{\Gamma,\zone{z}{\unit}}\to \Out(C).
\]

\medskip

\noindent\emph{Case 6: $\pi$ ends with $\mathrm{W}_z$.}
Then
\[
\pi\;=\;
\frac{\pi_0\colon \Gamma\vdash B}
     {\Gamma,\zone{z}{A}\vdash B}\;\mathrm{W}_z,
\]
with
\[
z\in \W_{\Lic}
\]
because the rule $\mathrm{W}_z$ is present in $\TZ_{\Sig_{\Lic}}$ exactly for such zones.
By the induction hypothesis,
\[
\sem{\pi_0}\colon \sem{\Gamma}\to \Out(B)
\]
lies in $\Diag_{\Lic}$.
Since $z\in\W_{\Lic}$, the discipline provides the licensed discard
\[
\omega_{z,\sem{A}}\colon J_z(\sem{A})\to \emptyset,
\]
and this morphism belongs to $\Diag_{\Lic}$ by \autoref{def:licensed-diagrams}(4).
By closure under tensor,
\[
\id_{\sem{\Gamma}}\tensor\omega_{z,\sem{A}}\colon
\sem{\Gamma}\tensor J_z(\sem{A})\to \sem{\Gamma}\tensor \emptyset
\]
belongs to $\Diag_{\Lic}$.
Composing with the right unitor
\[
\rho_{\sem{\Gamma}}\colon \sem{\Gamma}\tensor\emptyset\to\sem{\Gamma}
\]
and then with $\sem{\pi_0}$ yields a morphism
\[
\sem{\Gamma}\tensor J_z(\sem{A})\to \Out(B)
\]
in $\Diag_{\Lic}$.
Up to the canonical structural isomorphism between
\[
\sem{\Gamma,\zone{z}{A}}
\quad\text{and}\quad
\sem{\Gamma}\tensor J_z(\sem{A}),
\]
this gives the required interpretation of the conclusion.
By construction, the interpretation uses the designated family
\[
\omega_{z,-}.
\]

\medskip

\noindent\emph{Case 7: $\pi$ ends with $\mathrm{C}_z$.}
Then
\[
\pi\;=\;
\frac{\pi_0\colon \Gamma,\zone{z}{A},\zone{z}{A}\vdash B}
     {\Gamma,\zone{z}{A}\vdash B}\;\mathrm{C}_z,
\]
with
\[
z\in \C_{\Lic}.
\]
By the induction hypothesis,
\[
\sem{\pi_0}\colon \sem{\Gamma,\zone{z}{A},\zone{z}{A}}\to \Out(B)
\]
lies in $\Diag_{\Lic}$.
Since $z\in\C_{\Lic}$, the discipline provides the licensed diagonal
\[
\delta_{z,\sem{A}}\colon J_z(\sem{A})\to J_z(\sem{A})\tensor J_z(\sem{A}),
\]
and this morphism belongs to $\Diag_{\Lic}$ by \autoref{def:licensed-diagrams}(5).
Therefore
\[
\id_{\sem{\Gamma}}\tensor\delta_{z,\sem{A}}\colon
\sem{\Gamma}\tensor J_z(\sem{A})
\to
\sem{\Gamma}\tensor J_z(\sem{A})\tensor J_z(\sem{A})
\]
belongs to $\Diag_{\Lic}$.
Composing, as needed, with canonical associativity and symmetry isomorphisms to identify the codomain with
\[
\sem{\Gamma,\zone{z}{A},\zone{z}{A}},
\]
and then with $\sem{\pi_0}$, yields a morphism
\[
\sem{\Gamma,\zone{z}{A}}\to \Out(B)
\]
in $\Diag_{\Lic}$.
By construction, this interpretation uses the designated family
\[
\delta_{z,-}.
\]

\medskip

\noindent\emph{Case 8: $\pi$ ends with \emph{cut}.}
Then
\[
\pi\;=\;
\frac{\pi_1\colon \Gamma\vdash A
\qquad
\pi_2\colon \Delta,\zone{z}{A}\vdash B}
{\Gamma,\Delta\vdash B}\;\mathrm{cut}.
\]
By the induction hypothesis, there exist morphisms
\[
\sem{\pi_1}\colon \sem{\Gamma}\to \Out(A),
\qquad
\sem{\pi_2}\colon \sem{\Delta,\zone{z}{A}}\to \Out(B)
\]
in $\Diag_{\Lic}$.
Choose the ordered representative of the second premise so that
\[
\sem{\Delta,\zone{z}{A}}
\]
is canonically isomorphic to
\[
\sem{\Delta}\tensor J_z(\sem{A}).
\]
The input block
\[
\iota_{z,A}\colon \Out(A)\to J_z(\sem{A})
\]
belongs to $\Diag_{\Lic}$ by \autoref{def:licensed-diagrams}(2).
Hence
\[
\sem{\pi_1}\tensor \id_{\sem{\Delta}}\colon
\sem{\Gamma}\tensor\sem{\Delta}\to \Out(A)\tensor\sem{\Delta}
\]
belongs to $\Diag_{\Lic}$, and so does its composite with the symmetry
\[
\beta\colon \Out(A)\tensor\sem{\Delta}\to \sem{\Delta}\tensor\Out(A).
\]
By closure under tensor,
\[
\id_{\sem{\Delta}}\tensor \iota_{z,A}\colon
\sem{\Delta}\tensor\Out(A)\to \sem{\Delta}\tensor J_z(\sem{A})
\]
also lies in $\Diag_{\Lic}$.
Composing these morphisms, and then using the canonical isomorphism
\[
\sem{\Delta}\tensor J_z(\sem{A})\xrightarrow{\cong}\sem{\Delta,\zone{z}{A}},
\]
we obtain a morphism
\[
h\colon \sem{\Gamma}\tensor\sem{\Delta}\to \sem{\Delta,\zone{z}{A}}
\]
in $\Diag_{\Lic}$.
Finally, define
\[
\sem{\pi}:=\sem{\pi_2}\circ h.
\]
Then
\[
\sem{\pi}\colon \sem{\Gamma}\tensor\sem{\Delta}\to \Out(B)
\]
lies in $\Diag_{\Lic}$.
Up to the canonical structural isomorphism between
\[
\sem{\Gamma,\Delta}
\quad\text{and}\quad
\sem{\Gamma}\tensor\sem{\Delta},
\]
this yields the required interpretation
\[
\sem{\Gamma,\Delta}\to \Out(B).
\]

\medskip

These are all the rules of $\TZ_{\Sig_{\Lic}}$. 
In every case, the interpretation lies in $\Diag_{\Lic}$.
This completes the induction.\qedhere
\end{proof}

The theorem shows that every derivation is interpreted by a diagram assembled from the architectural data, the coherence isomorphisms of the monoidal structure, and the designated structural families. 
In particular, the structural rules do not receive arbitrary semantic witnesses: they are interpreted by the very discards and diagonals licensed by the architecture.

\begin{cor}\label{cor:resource-reading}
The extracted proof theory is sound exactly for the structural behaviour explicitly licensed by the coherently disciplined architecture.
\end{cor}

\begin{proof}
By \autoref{thm:soundness}, every derivation in $\TZ_{\Sig_{\Lic}}$ is interpreted by a morphism in $\Diag_{\Lic}$, and the rules $\mathrm{W}_z$ and $\mathrm{C}_z$ are interpreted using the designated families
\[
\omega_{z,-}
\qquad\text{and}\qquad
\delta_{z,-},
\]
respectively.

Conversely, by \autoref{thm:recovery}, the sets of zones in which weakening and contraction occur in the extracted signature are exactly the zones in which the discipline $\Lic$ supplies licensed discard and diagonal families. 
Therefore the structural rules present in the extracted proof theory are precisely those corresponding to structurally licensed architectural behaviour.

Hence the soundness theorem identifies the proof theory with exactly the structural behaviour explicitly licensed by the categorized architecture.\qedhere
\end{proof}

\begin{exa}\label{exa:three-zones}
We spell out the extraction process on a minimal three-zone discipline.

Let
\[
Z=\{p,r,\ell\},
\]
and let the preorder on $Z$ be discrete, so that
\[
z\preceq z' \quad\Longleftrightarrow\quad z=z'.
\]
Thus there are no nontrivial zone coercions: for each $z\in Z$ and each $A\in\Ob(\Ctx)$, the only coercion is
\[
\chi_{z,z,A}=\id_{J_z(A)}.
\]

We interpret the three zones as follows:
\begin{itemize}
\item $p$ is a \emph{persistent} zone;
\item $r$ is a \emph{relevant} zone;
\item $\ell$ is a \emph{linear} zone.
\end{itemize}
The discipline is specified by
\[
\W_{\Lic}=\{p\},
\qquad
\C_{\Lic}=\{p,r\}.
\]
Accordingly, for every object $A\in\Ob(\Ctx)$, the licensed structural family contains:
\begin{itemize}
\item a discard
\[
\omega_{p,A}\colon J_p(A)\to \emptyset;
\]
\item diagonals
\[
\delta_{p,A}\colon J_p(A)\to J_p(A)\tensor J_p(A),
\qquad
\delta_{r,A}\colon J_r(A)\to J_r(A)\tensor J_r(A);
\]
\item no discard in zone $r$;
\item neither discard nor diagonal in zone $\ell$.
\end{itemize}

Since the preorder is discrete, upward closure imposes no further conditions. The reconstructed structural classes therefore satisfy
\[
\W^{\disc}_{\Lic}=\{p\},
\qquad
\C^{\diag}_{\Lic}=\{p,r\}.
\]
Hence the extracted signature is
\[
\Sig_{\Lic}=(\{p,r,\ell\},=,\{p\},\{p,r\}).
\]

The corresponding extracted calculus has:
\begin{itemize}
\item a weakening rule only in zone $p$,
\[
\frac{\Gamma\vdash B}{\Gamma,\zone{p}{A}\vdash B}\;\mathrm{W}_p;
\]
\item contraction rules in zones $p$ and $r$,
\[
\frac{\Gamma,\zone{p}{A},\zone{p}{A}\vdash B}{\Gamma,\zone{p}{A}\vdash B}\;\mathrm{C}_p,
\qquad
\frac{\Gamma,\zone{r}{A},\zone{r}{A}\vdash B}{\Gamma,\zone{r}{A}\vdash B}\;\mathrm{C}_r;
\]
\item no weakening rule in zone $r$;
\item neither weakening nor contraction in zone $\ell$.
\end{itemize}

Thus the proof theory mirrors the intended architectural behaviour:
persistent information may be copied and discarded, relevant information may be copied but not silently erased, and linear information must be used without free copying or erasure.

For instance, the sequent
\[
\zone{p}{X}\vdash \unit
\]
is derivable, since one may first derive
\[
\emptyset\vdash \unit
\]
by $\unit\mathrm{R}$ and then apply $\mathrm{W}_p$. By contrast, the calculus contains no rule $\mathrm{W}_r$, so the same derivation pattern is unavailable in zone $r$.

Likewise, from the axiom
\[
\zone{r}{X}\vdash X
\]
one obtains
\[
\zone{r}{X},\zone{r}{X}\vdash X
\]
by applying $\mathrm{W}_r$ is \emph{not} possible, since $r\notin\W_{\Lic}$. However, if a derivation already contains two copies of $\zone{r}{X}$, then $\mathrm{C}_r$ permits them to be contracted to one. In this way the relevant zone supports duplication control without unrestricted erasure.

This example makes explicit the mechanism developed in the paper: one starts from a categorized structural discipline, reconstructs from it the weakening and contraction classes, and thereby obtains a proof system whose structural rules are exactly those licensed by the architecture.

\end{exa}

The previous example isolates the structural mechanism abstractly. We now instantiate the same mechanism on a minimal neural architecture.

\begin{exa}[A two-block neural architecture and its extracted logic]\label{exa:two-block-neural}
We now give a concrete example showing the passage
\[
\text{neural architecture}\;\longrightarrow\;\text{categorisation}\;\longrightarrow\;\text{subexponential signature}\;\longrightarrow\;\text{proof theory}.
\]

Let $\Ctx=\mathbf{FinVect}_{\mathbb R}$, and let
\[
Z=\{p,r,\ell\},
\]
where $p$ is a persistent zone, $r$ is a relevant zone, and $\ell$ is a linear zone.
Assume that the preorder on $Z$ is discrete.

Choose finite-dimensional real vector spaces
\[
M=\mathbb R^m,\qquad
R=\mathbb R^r,\qquad
X=\mathbb R^d,\qquad
H=\mathbb R^h,\qquad
Y=\mathbb R^k.
\]

We interpret:
\begin{itemize}
\item $M$ as persistent memory;
\item $R$ as retrieved contextual information;
\item $X$ as the current observation;
\item $H$ as a hidden representation;
\item $Y$ as the output space.
\end{itemize}

Consider first a neural block
\[
f\colon M\times R\times X\longrightarrow H,
\qquad
f(m,r,x)=\sigma(W_Mm+W_Rr+W_Xx+b),
\]
where
\[
W_M\colon M\to H,\qquad
W_R\colon R\to H,\qquad
W_X\colon X\to H
\]
are linear maps, $b\in H$, and $\sigma\colon H\to H$ is a componentwise activation.

This determines a parametrised resource block with trivial parameter object
\[
(1,f)\colon
\ctx{p,M}\tensor\ctx{r,R}\tensor\ctx{\ell,X}\longrightarrow \ctx{\ell,H},
\]
hence a morphism
\[
[\![1,f]\!]\colon
J_p(M)\tensor J_r(R)\tensor J_\ell(X)\longrightarrow J_\ell(H)
\]
in the architectural category.

Next, consider a second neural block
\[
g\colon H\times M\longrightarrow Y,
\qquad
g(h,m)=\tau(V_Hh+V_Mm+c),
\]
where
\[
V_H\colon H\to Y,\qquad
V_M\colon M\to Y
\]
are linear maps, $c\in Y$, and $\tau\colon Y\to Y$ is another componentwise activation.

This yields a second morphism
\[
[\![1,g]\!]\colon
J_\ell(H)\tensor J_p(M)\longrightarrow \Out(Y).
\]

The resulting architecture uses the persistent memory input twice: once in the first block and once again in the second. To model this categorically, we impose a coherent zone discipline by declaring
\[
\W_{\Lic}=\{p\},
\qquad
\C_{\Lic}=\{p,r\}.
\]
Thus:
\begin{itemize}
\item the persistent zone $p$ admits both discard and contraction;
\item the relevant zone $r$ admits contraction but not discard;
\item the linear zone $\ell$ admits neither discard nor contraction.
\end{itemize}

The licensed structural family therefore contains
\[
\omega_{p,A}\colon J_p(A)\to \emptyset
\]
for every $A\in\Ob(\Ctx)$, and
\[
\delta_{p,A}\colon J_p(A)\to J_p(A)\tensor J_p(A),
\qquad
\delta_{r,A}\colon J_r(A)\to J_r(A)\tensor J_r(A)
\]
for every $A\in\Ob(\Ctx)$, but contains neither
\[
\omega_{r,A},\ \omega_{\ell,A},\ \delta_{\ell,A}.
\]

The extracted subexponential signature is therefore
\[
\Sig_{\Lic}=(\{p,r,\ell\},=,\{p\},\{p,r\}).
\]
The associated proof system has:
\begin{itemize}
\item weakening only in zone $p$;
\item contraction in zones $p$ and $r$;
\item no weakening in zones $r$ and $\ell$;
\item no contraction in zone $\ell$.
\end{itemize}

The logical reading is immediate:
\begin{itemize}
\item persistent memory behaves as reusable and discardable context;
\item retrieved context behaves as reusable but non-discardable context;
\item the current observation behaves linearly.
\end{itemize}

The categorical composition of the two neural blocks uses the licensed diagonal
\[
\delta_{p,M}\colon J_p(M)\to J_p(M)\tensor J_p(M)
\]
to duplicate the persistent input. Up to canonical associativity and symmetry isomorphisms, one obtains the composite
\[
J_p(M)\tensor J_r(R)\tensor J_\ell(X)
\xrightarrow{\ \delta_{p,M}\tensor \id \tensor \id\ }
J_p(M)\tensor J_p(M)\tensor J_r(R)\tensor J_\ell(X)
\]
followed by
\[
[\![1,f]\!]\tensor \id_{J_p(M)}
\]
and then by
\[
[\![1,g]\!].
\]
This yields a morphism
\[
J_p(M)\tensor J_r(R)\tensor J_\ell(X)\longrightarrow \Out(Y)
\]
in the licensed diagram class $\Diag_{\Lic}$.

The example also illustrates the soundness theorem at the level of individual rules. 
The rule
\[
\frac{\Gamma\vdash B}{\Gamma,\zone{p}{A}\vdash B}\;\mathrm{W}_p
\]
is interpreted in $\Diag_{\Lic}$ by tensoring the interpretation of the premise with the designated discard
\[
\omega_{p,\sem{A}}\colon J_p(\sem{A})\to \emptyset,
\]
followed by the canonical right unitor
\[
\rho_{\sem{\Gamma}}\colon \sem{\Gamma}\tensor \emptyset\to \sem{\Gamma}.
\]
Thus the soundness interpretation of $\mathrm{W}_p$ is the composite
\[
\sem{\Gamma}\tensor J_p(\sem{A})
\xrightarrow{\ \id_{\sem{\Gamma}}\tensor \omega_{p,\sem{A}}\ }
\sem{\Gamma}\tensor \emptyset
\xrightarrow{\ \rho_{\sem{\Gamma}}\ }
\sem{\Gamma}
\xrightarrow{\ \sem{\pi_0}\ }
\Out(B),
\]
where $\sem{\pi_0}\colon \sem{\Gamma}\to \Out(B)$ interprets the premise derivation.

Likewise, the rule
\[
\frac{\Gamma,\zone{p}{A},\zone{p}{A}\vdash B}{\Gamma,\zone{p}{A}\vdash B}\;\mathrm{C}_p
\]
is interpreted by tensoring with the designated diagonal
\[
\delta_{p,\sem{A}}\colon J_p(\sem{A})\to J_p(\sem{A})\tensor J_p(\sem{A}),
\]
so that its semantic action is given by
\[
\sem{\Gamma}\tensor J_p(\sem{A})
\xrightarrow{\ \id_{\sem{\Gamma}}\tensor \delta_{p,\sem{A}}\ }
\sem{\Gamma}\tensor J_p(\sem{A})\tensor J_p(\sem{A})
\longrightarrow \Out(B),
\]
where the second arrow is the interpretation of the premise, up to the canonical associativity and symmetry isomorphisms used to identify the codomain with
\[
\sem{\Gamma,\zone{p}{A},\zone{p}{A}}.
\]

Similarly, the rule
\[
\frac{\Gamma,\zone{r}{A},\zone{r}{A}\vdash B}{\Gamma,\zone{r}{A}\vdash B}\;\mathrm{C}_r
\]
is interpreted using the designated diagonal
\[
\delta_{r,\sem{A}}\colon J_r(\sem{A})\to J_r(\sem{A})\tensor J_r(\sem{A}).
\]

By contrast, there is no rule
\[
\mathrm{W}_r
\]
in the extracted calculus, and this absence is reflected semantically by the absence of a designated discard
\[
\omega_{r,\sem{A}}\colon J_r(\sem{A})\to\emptyset
\]
from the licensed family. Likewise, there is no rule
\[
\mathrm{C}_\ell,
\]
because the discipline supplies no diagonal
\[
\delta_{\ell,\sem{A}}\colon J_\ell(\sem{A})\to J_\ell(\sem{A})\tensor J_\ell(\sem{A}).
\]

Hence the proof theory mirrors the structural behaviour of the architecture exactly: the logic is not imposed externally on the network, but is read from the way the categorized architecture permits or forbids reuse and erasure of its different classes of inputs.

This example may be read as a minimal witness that the passage
\[
\text{neural architecture}\to\text{categorisation}\to\text{logic}
\]
is non-vacuous already at the level of small compositional networks.
\end{exa}
\section{Positioning with respect to related work}\label{sec:related}

The paper sits at the intersection of three literatures; \autoref{tab:related-comparison} summarises the contrast.

First, compositional semantics for learning systems studies how parametrised models and training procedures compose categorically \cite{FongSpivakTuyeras2019}. Categorical deep learning pushes that line further by proposing a single algebraic language for architectural specification and implementation \cite{GavranovicEtAl2024}. Our work is compatible with that programme but has a different target: we do not seek a universal algebra of all architectures. Instead, we isolate one particular invariant of architectural organization, namely differentiated structural behaviour of contexts and memory.

Second, subexponential proof theory and its categorical semantics provide the immediate
logical background for resource-indexed structural rules. More broadly, categorical semantics
for linear logic has a substantial history; see, for example,
\cite{See89,Bar91,HS03,Mel09}. Algorithmic and multimodal uses of subexponentials have
been studied for some time \cite{NigamMiller2009,XavierOlartePimentel2022}, and recent work develops categorical
models for intuitionistic linear logic with subexponentials \cite{Rogozin2025}. Our paper does not compete with that general theory. Rather, it contributes a complementary direction of travel: starting from categorized architectural data, we extract the subexponential signature that best matches the architectural discipline.

Third, there is growing interest in connecting logic and neural architecture design more directly \cite{SinghVasicKhurshid2020}. In that broader landscape, the present article should be read as a foundational bridge. We do not start with a logic and synthesise a network. We start with a family of zone-labelled architectural blocks and ask what logic their categorized resource discipline validates.

\begin{table}[t]
\centering
\begin{tabular}{@{}p{0.22\textwidth}p{0.21\textwidth}p{0.20\textwidth}p{0.29\textwidth}@{}}
\toprule
\textbf{Line of work} & \textbf{Starting point} & \textbf{Primary output} & \textbf{Difference with the present work} \\
\midrule
Categorical deep learning &
Neural architectures and their compositional structure &
Categorical semantics of parametrised learning systems &
The present work does not stop at a categorical semantics of architectures: it reconstructs a subexponential logical discipline from a categorized architectural resource policy. \\[0.5em]

Categorical semantics for subexponential linear logics &
A logic or proof system already given in advance &
A categorical model validating that logic &
Here the direction is reversed: the subexponential signature is not presupposed, but extracted from a coherently disciplined architectural category. \\[0.5em]

Logic-guided or logic-constrained neural design &
A logical theory, specification, or constraint language &
A neural architecture constrained or guided by that logic &
Here logic is not used to design the architecture from outside; rather, a logical system is read from the architectural treatment of zones, reuse, and discard. \\[0.5em]

This paper &
A categorized architecture equipped with a coherent zone discipline &
An extracted subexponential signature, a tensorial proof theory, and a sound categorical interpretation &
The contribution is an architecture-to-logic pipeline: categorized structural permissions determine the extracted logical rules. \\
\bottomrule
\end{tabular}
\caption{Positioning of the present work relative to nearby research directions.}
\label{tab:related-comparison}
\end{table}

\section{Conclusion}\label{sec:conclusion}

The article has established a first rigorous fragment of the architecture-to-category-to-logic programme. The main message is that one can begin from a categorized family of resource-sensitive architectural blocks, read off a subexponential signature from the licensed structural behaviour, and obtain a sound proof theory for the resulting tensorial fragment. The direction of explanation matters: the logic is extracted from the categorized architecture rather than fixed in advance.

Three extensions are immediate.
\begin{enumerate}
\item Enrich the categorized architecture with the additional structure needed to interpret linear implication, most likely through a satisfactory internal hom or a linear--nonlinear adjunction \cite{Benton1995,Rogozin2025}.
\item Replace the tensorial fragment by a fuller subexponential proof system using the extracted preorder more actively.
\item Study concrete architecture families whose memory, retrieval, or tool-use policies induce nontrivial zone disciplines.
\item The two-block neural example shows that the extraction programme is already non-vacuous
at the level of small compositional architectures: the logic is read from architectural
resource policies rather than imposed externally.
\end{enumerate}

Even in its present form, however, the paper already yields nontrivial theorems: cut elimination for the extracted calculus, structural adequacy for the licensed permissions, and soundness in the original categorized architecture. Those results make the proposal mathematically substantive enough to stand as a first paper in its own right.

\bibliographystyle{alphaurl}
\bibliography{refs}

\end{document}